\documentclass[11pt]{article}
\usepackage{graphicx,amsthm,graphicx,fancyhdr,mathrsfs}
\usepackage{amsfonts}
\usepackage{amsmath}
\usepackage{amssymb}
\usepackage{url}
 \usepackage[colorlinks=true,citecolor=blue]{hyperref}
\usepackage{fancyhdr}
\usepackage{indentfirst}
\usepackage{enumerate}
\usepackage{amsthm}
\usepackage{color}
\usepackage{comment}
\usepackage{dsfont}
\usepackage{natbib}
\usepackage[misc]{ifsym}
\usepackage[toc,page]{appendix}
\usepackage{multirow}
\usepackage{float}
\usepackage{bm}
\usepackage{booktabs}
\usepackage{todonotes}
\usepackage[font=scriptsize]{caption} 

\def\d{\,\mathrm{d}}

\newcommand{\VaR}{\mathrm{VaR}}

\newcommand{\F}{\mathcal{F}}

\newcommand{\X}{\mathcal{X}}

\newcommand{\E}{\mathbb{E}}
\newcommand{\R}{\mathbb{R}}

\newcommand{\M}{\mathcal{M}}

\newcommand{\p}{\mathbb{P}}

\newcommand{\id}{\mathds{1}}

\renewcommand{\ge}{\geqslant}
\renewcommand{\le}{\leqslant}
\renewcommand{\geq}{\geqslant}
\renewcommand{\leq}{\leqslant}

\theoremstyle{plain}
\newtheorem{theorem}{Theorem}
\newtheorem{corollary}{Corollary}
\newtheorem{lemma}{Lemma}
\newtheorem{proposition}{Proposition}
\theoremstyle{definition}

\newtheorem{remark}{Remark}

\usepackage{tikz}

\renewcommand{\cite}{\citet}

\makeatletter
\DeclareRobustCommand{\bsquare}{%
	\mathop{\vphantom{\sum}\mathpalette\bigstar@\relax}\slimits@
}

\newcommand{\bigstar@}[2]{%
	\vcenter{%
		\sbox\z@{$#1\sum$}%
		\hbox{\resizebox{.9\dimexpr\ht\z@+\dp\z@}{!}{$\m@th\dsquare$}}%
	}%
}
\makeatother

\usepackage[onehalfspacing]{setspace}

\newcommand{\dsquare}{\mathop{  \square} \displaylimits}

\begin{document}

	\title{Optimal insurance design with Lambda-Value-at-Risk}

	\author{Tim J.~Boonen\thanks{Department of Statistics and Actuarial Science, School of Computing and Data Science, The University of Hong Kong, Hong Kong, China \Letter~\texttt{tjboonen@hku.hk}} \and Yuyu Chen\thanks{Department of Economics, University of Melbourne, Melbourne, Victoria, Australia \Letter~\texttt{yuyu.chen@unimelb.edu.au}}\and Xia Han\thanks{School of Mathematical Sciences and LPMC, Nankai University, Tianjin, China \Letter~\texttt{xiahan@nankai.edu.cn}}\and Qiuqi Wang\thanks{Maurice R.~Greenberg School of Risk Science, Georgia State University, Atlanta, Georgia, U.S.A. \Letter~\texttt{qwang30@gsu.edu}}}
	\maketitle 
	
	\begin{abstract}
		This paper explores optimal insurance solutions based on the Lambda-Value-at-Risk ($\Lambda\VaR$). {Using the expected value premium principle, we first analyze a stop-loss indemnity and provide a closed-form expression for the deductible parameter. A necessary and sufficient condition for the existence of a positive and finite deductible is also established. We then generalize the stop-loss indemnity and show that, akin to the VaR model,   a limited stop-loss indemnity remains optimal within the $\Lambda\VaR$ framework.} Further,  we examine the use of  $\Lambda'\VaR$ as a premium principle  and show that  full or no insurance is optimal. We also identify {that} a limited loss indemnity is optimal  when $\Lambda'\VaR$ is solely used to determine the risk-loading in the premium principle. 
 Additionally, we investigate the impact of model uncertainty, particularly in scenarios where the loss distribution is unknown but lies within a specified uncertainty set. Our findings suggest that a limited stop-loss indemnity is optimal when the uncertainty set is defined using a likelihood ratio. Meanwhile, {when only the first two moments of the loss distribution are available,}
 we provide a closed-form expression for the optimal deductible in a stop-loss indemnity.
		
		~
		
	\end{abstract} 
	
	~

	\section{Introduction}
 \label{sec:intro}
Optimal insurance is the problem of finding the optimal insurance contract. A decision maker (DM) seeks to transfer a portion of his loss to an insurer and will pay a premium for this coverage.  The premium is traditionally calculated as the expected value of the insured risk, plus a risk loading. The literature {on optimal insurance} originated from the early works of \cite{B60} and \cite{A63}, which focused on mean-variance and expected utility preferences, respectively.  {Many authors have extended these settings by considering alternative objective functions such as behavioral economic preferences and risk measures for evaluating future losses after insurance. For instance, \cite{CT07}, \cite{CTWZ08}, and  \cite{CT11} study the optimal insurance problem by minimizing the Value-at-Risk (VaR) and the Conditional Value-at-Risk (CVaR) of the DM's total risk exposure. Subsequently,  \cite{CYW13} and \cite{A15} generalize these settings with distortion risk measures, which include VaR and CVaR as special cases}. 
 
{In this paper, we explore optimal insurance  problems using the Lambda-Value-at-Risk ($\Lambda\VaR$) as the DM's risk measure. The $\Lambda\VaR$ is introduced by \cite{FMP14},  and has recently gained academic interest as a  generalization of the well-known VaR.  Unlike $\VaR$, which uses a fixed confidence level, $\Lambda\VaR$ incorporates a function $\Lambda$ as its confidence level that varies with the underlying loss realizations.  
This feature provides greater flexibility, allowing one to adjust  $\Lambda$ to  either increase  sensitivity to extreme losses or  focus more on the most likely outcomes.  Similar to VaR, and in contrast to CVaR,  $\Lambda\VaR$ is always finite (except for probability levels $0$ and $1$), even for heavy-tailed losses with infinite expectations.} While $\Lambda\VaR$ is a monotonic risk measure, it differs from VaR in that it is not cash- or comonotonic additive.  A characterization of $\Lambda\VaR$ is provided by \cite{BP22}, who  formalize a crucial property called locality. This,  together with monotonicity, normalization, and weak semi-continuity, uniquely characterizes $\Lambda\VaR$  with a non-increasing function $\Lambda(x)$ on $\R$.  $\Lambda\VaR$ has also been studied from various  perspectives, including robustness, elicitability and consistency \citep[see][]{BPR17}, estimation and backtesting \citep[see][]{HMP18, CP18}, quasi-convexity, cash subadditivity and quasi-star-shapedness \citep[see][]{HWWX25}, robust formulations \citep[see][]{HL24}, and capital allocation with the Euler rule \citep[see][]{IPP22}.


An optimal {insurance} problem based on $\Lambda\VaR$ with the  expected value premium is first {considered} by \cite{BBB23},  as an application to their main theory on $\Lambda$-fixed points. Their result assumes finite support of the function $\Lambda$ and solves the optimization problem  using a linear programming method. {As a significant improvement to the existing discussion, we solve the $\Lambda\VaR$-based optimal insurance problem with general decreasing $\Lambda$ functions and loss distributions by more concise theoretical proofs. Moreover, various assumptions of indemnities and premium principles are considered. } To the best of our knowledge, this work is the first to comprehensively study the optimal (re)insurance design problem with $\Lambda\VaR$ {in} a general setup.

{ We first focus on the problem with the expected premium principle. Assuming that the admissible set of indemnities is of stop-loss types, we present a closed-form expression for the optimal deductible. Moreover, we identify a necessary and sufficient condition for the existence of a positive and finite optimal deductible. Compared with the corresponding condition in the VaR-based optimal insurance problem of \cite{CT07}, our condition is less stringent. Next, we consider the more general class of incentive-compatible indemnity functions \citep{HMS83}, denoted by $\F$. Within the class $\F$, \cite{CT11} show that a limited stop-loss indemnity function is optimal for a  VaR-based optimal insurance problem.
We extend this finding to  $\Lambda\mathrm{VaR}$, with which the DM can adjust the
upper limit of the ceded loss function based on various model parameters and insurance premiums. 

Following this, we proceed to consider different premium principles with indemnities in $\F$. If a $\Lambda'\mathrm{VaR}$ premium principle is applied, we demonstrate that full or zero insurance is the optimal indemnity, with the optimal structure determined solely by the ordering of the $\Lambda\mathrm{VaR}$ and $\Lambda'\mathrm{VaR}$ of the total loss. In the case of a more general $\Lambda'\mathrm{VaR}$-based premium principle, where $\Lambda'\mathrm{VaR}$ is used to determine a risk loading on top of the expected value of losses, we show that a limited loss indemnity is optimal.}

The impact of model uncertainty in optimization based on risk measures has been increasingly recognized in academic research, leading to extensive studies on the topic. For instance,  distributionally robust portfolio optimization has been explored by \cite{HZFF10},  \cite{NSU10}, \cite{GX13},
\cite{BCZ22}, \cite{SZ23} and \cite{PWW20}.  More recently, model uncertainty has also been extensively studied in the context of insurance and reinsurance. 
\cite{ABCHK17} study the robust and Pareto optimality of insurance contracts; \cite{LM22} propose a robust reinsurance strategy  under moment constraints;  \cite{BBG23} and \cite{BJ24} examine robust insurance problems  by assuming  that the uncertainty set is defined using Wasserstein distance; \cite{CLY24} study the stop-loss and limited loss strategies when the 
uncertainty set  is characterized by mean, variance, and  Wasserstein distance;
\cite{LLS24} propose a non-cooperative optimal reinsurance problem incorporating likelihood ratio uncertainty. 

{This paper presents an analysis of the impact of model uncertainty on the $\Lambda$VaR optimal insurance problem.} We assume that the underlying loss distribution is unknown, but it is an element of a so-called uncertainty set, {which can be thought of as a collection of plausible loss distribution functions}. If such a set is generated by a ball centered around the probability measure used for pricing, measured using a likelihood ratio, then  a {limited} stop-loss indemnity is shown to be optimal {within $\F$}. When uncertainty is characterized by knowledge of the first two moments of the loss variable, and we restrict our focus to stop-loss indemnities,  a closed-form expression for the optimal deductible is provided. {In general}, we find that when the uncertainty set is defined by any collection of loss distributions, the robust $\Lambda\mathrm{VaR}$ optimization problem can be transformed into a robust VaR optimization problem over the same uncertainty set. Subsequently, the optimal probability level for $\Lambda\mathrm{VaR}$ can be determined based on the results from the VaR problem. This observation is consistent with \cite{HL24}, which does not incorporate an insurance strategy.

The findings of this paper are valuable to both academics and practitioners, as they provide insights into the optimal indemnities associated with a new risk measure, extending beyond the well-studied class of distortion risk measures. Similar to distortion risk measures, we show that optimal indemnities typically exhibit a piecewise-linear structure, with proportional or co-insurance generally being suboptimal. Specifically, our results complement the existing literature on optimal (re)insurance with risk measures, including but not limited to \cite{ABCHK17}, \cite{AB18}, \cite{TWWZ20}, \cite{ASL21}, and \cite{BJ24}.

{We note that the optimal insurance problems considered in this paper are different from the $\Lambda$VaR-based risk sharing problems in the literature \citep[e.g.,][]{L24,LTW24,XH24}.} 
In optimal risk sharing problems, the aim is to determine an optimal set of random variables that aggregate to the total risk in the market,  while minimizing the sum of risk measures. 
{This objective differs from that of optimal insurance problems, where the goal is to minimize a risk measure for a single agent. Moreover, in optimal insurance problems, the admissible set of indemnities is typically constrained to a specific class of non-decreasing 1-Lipschitz functions, as discussed in \cite{A15}.}
{ This assumption implies that the retained and ceded losses in optimal insurance problems are comonotonic.  In sharp contrast,  optimal risk sharing allocations are  not comonotonic as shown by  \cite{L24} and \cite{LTW24}.} Additionally, in optimal insurance, the insurer's objectives are summarized through a premium principle, which may include a risk loading.   

The paper is organized as follows. Section \ref{sec:setup} introduces $\Lambda\mathrm{VaR}$ and the optimal insurance problem. Section \ref{sec:EVpremium} examines optimal insurance indemnities under the expected value premium principle. Section \ref{sec:LambdaVaR} investigates optimal insurance when the premium principle {incorporates $\Lambda'\mathrm{VaR}$}. Section \ref{sec:model-uncertainty} discusses optimal insurance indemnities under model uncertainty. Section \ref{sec:examples} presents some numerical examples, and Section \ref{sec:conc} concludes.

\section{Problem formulation}
\label{sec:setup}
Let $(\Omega, \mathcal{G}, \mathbb{P})$ be a probability space.  Each random variable represents a random risk that is realized at a well-defined future period. 
 Throughout the paper,  ``increasing" and ``decreasing" are in the  non-strict (weak) sense.  Let  $\X$ be a convex cone of random variables.   For any $Z \in \X$, the cumulative distribution function (CDF) associated with $Z$ is denoted by $F_Z$.  For a function $f:\R\to\R$ and $x\in\R$, we write $\lim_{t\uparrow x}f(t)=f(x-)$.

 A random vector $(Z_1,\dots,Z_n)$ is comonotonic
if there exists a random variable $Z$ and increasing functions $f_1,\dots,f_n$ on $\R$ such that $Z_i=f_i(Z)$ a.s.~for every $i=1,\dots,n$.  We define the following properties for a mapping $\rho: \mathcal{X} \rightarrow \mathbb{R}$: 
\begin{itemize}
\item[{[A1]}] (Monotonicity)  $\rho(Y)\le \rho(Z)$ for all $Y,Z \in \mathcal X$ with $Y\le Z$;
\item[{[A2]}] (Cash additivity) $\rho (Y+c)= \rho (Y)+c$ for all  $c\in\R$;
\item[{[A3]}] (Comonotonic additivity)  $\rho(Y+Z)=\rho(Y)+\rho(Z)$ whenever $Y$ and $Z$ are comonotonic.
\end{itemize} Using the standard terminology in \cite{FS16}, a mapping \(\rho: \mathcal{X} \rightarrow \mathbb{R}\) is called a \emph{monetary risk measure} if it satisfies [A1] and [A2].  

In this paper, we consider an application of  $\Lambda\VaR$ in optimal insurance design problems.  {Suppose that a decision maker (DM) faces a non-negative random loss $X\in\X_+$, where $$\X_+=\{X\in\X\mid X\geq0,~\p(X<\infty)=1\}.$$
For generality, we do not assume that $X$ has a finite expectation as infinite-mean models are particularly 
 useful in modeling large insurance claims and catastrophic losses; see, e.g., \cite{CW25}.}
 For $\Lambda: \R_+ \rightarrow [0,1]$, the risk measure \emph{Lambda-Value-at-Risk} \citep[$\Lambda\VaR$, see][]{FMP14} is defined as
\begin{equation}\label{eq:lambdadef}
\Lambda\VaR(X) = \inf \{x\in\R_+: \mathbb P(X\le x) \ge \Lambda(x)\},~~~X\in \X_+.
\end{equation}
{ If $\Lambda=\alpha\in[0,1]$, then $\Lambda\VaR$ boils down to the \emph{Value-at-Risk (VaR)} at level   $\alpha$ given by 
\begin{equation*}\label{eq:VaR}\VaR_\alpha(X) =\inf \left\{x\in\R_+: \mathbb P(X\le x) \ge \alpha \right\},~~~X\in \X_+.\end{equation*}
By convention, we use $ \inf \emptyset = \infty$ and $\sup \emptyset = 0.$}
It is well known that $\VaR$ satisfies all three properties [A1]--[A3]. However, \(\Lambda\VaR\) satisfies [A1] but generally does not satisfy [A2] or [A3], and thus does not belong to the class of monetary risk measures.
 For more properties of $\Lambda\VaR$, we refer to \cite{BP22}, \cite{HL24} and \cite{XH24}.

{  
We focus on a special class of functions  $\mathcal{D}$ defined below
$$\mathcal{D}=\{\Lambda: \mathbb{R}_+ \to [0,1]\mid\mbox{$\Lambda$ is decreasing with $0<\Lambda(0)<1$}\}.$$}
In this setup, {higher losses are tolerated with a higher probability by the DM}. Although \cite{FMP14} mainly study increasing $\Lambda$,
the recent works of \cite{BPR17} and \cite{BP22} have shown that using a decreasing $\Lambda$ leads to many advantages, including robustness, elicitability, and an axiomatic characterization. In addition, $\Lambda\VaR$ with decreasing $\Lambda$ functions also satisfy cash-subadditivity and quasi-star-shapedness; see \cite{KR09} and \cite{HWWX25}.  For all these reasons, we assume that  $\Lambda$ is  a decreasing function in this paper.  

Under an insurance contract, the insurer agrees to cover part of the loss $X$ and requires a premium in return. The function $f:  \R_+\to\R_+$   is commonly described as the indemnity or the ceded loss function, while $R(x) \triangleq x-f(x)$ is known as the retained loss function.  
To prevent the potential ex post moral hazard, where the DM might be incentivized to manipulate the losses, we follow the literature, such as \cite{HMS83} and \cite{CD03}, and impose the incentive compatibility condition on the indemnity functions. That is,  we consider insurance contracts $f \in \mathcal{F}$, where \begin{equation}\label{eq:I}
\F:=\left\{f: \R_+\to\R_+ \mid f(0)=0 \text { and } 0 \leqslant f(x)-f(y) \leqslant x-y, \text { for all } 0 \leqslant y \leqslant x\right\} . 
\end{equation}
Obviously, for any $f \in \F, f(x)$ and $x-f(x)$ are increasing in $x$,  and any $f\in\F$ is 1-Lipschitz continuous.\footnote{{For $K>0$, a function $f:\R\to\R$ is \emph{$K$-Lipschitz continuous} if $|f(x)-f(y)|\le K|x-y|$ for all $x,y\in\R$. We call $K$ the \emph{Lipschitz constant}.}} The assumption that {$f\in\F$} is common in the literature; see e.g., \cite{A15} and the review paper by \cite{CC20}. { This class is  rich enough and  includes many commonly used indemnity functions, such as the stop-loss function $f(x)=(x-d)_{+}$ for some $d \geq 0$ and the quota-share function $f(x)=q x$ for some $q \in[0,1]$. 
Since Lipschitz-continuous functions are absolutely continuous,  they are almost everywhere differentiable. It  follows directly  that  $f$ can be written as the integral of its derivatives essentially bounded by $1$, and $\F$ can be represented as}
\begin{equation}\label{eq:marginal}
\F=\left\{f:  \R_+\to\R_+  \mid f(x)=\int_0^x q(t) \d t, ~0 \leq q \leq 1,~{ x\in\R_+}\right\},
\end{equation}
where $q$ is called the  \emph{marginal indemnity function} \citep{A15,ZWTA16}.

For a given $f \in \F$, the insurer prices ceded losses {using an exogenously given premium principle $\Pi:\X\mapsto\R_+$. The premium is thus given by $\Pi(f(X))$, which is non-negative.}  The risk exposure of the DM after purchasing insurance is given by
 \begin{equation}\label{eq:T_f}T_f(X)= X-f(X)+\Pi(f(X)).\end{equation}
We consider the following optimization problem:
\begin{equation}\label{eq:obj_general}
    \inf_{f\in\F}\Lambda\VaR(T_f(X)),
\end{equation}
where we aim to find an optimal ceded loss function such that the risk of the DM is minimized when evaluated by $\Lambda\VaR$. The following lemma is useful to solve the optimization problem \eqref{eq:obj_general}.

\begin{lemma}[Theorem 3.1 of \citealp{HWWX25}]\label{lem:representation}
    For $X\in\X_+$ and $\Lambda\in\mathcal D$,   
the risk measure $\Lambda\VaR$ in \eqref{eq:lambdadef} has the representation
\begin{equation}
\label{eq:lambdarep}
\Lambda\VaR(X) = \inf_{x\in\R_+} \left\{\VaR_{\Lambda(x)} (X) \vee x\right\} 
= \sup_{x\in\R_+} \left\{\VaR_{\Lambda(x)} (X) \wedge x\right\}.
\end{equation}
\end{lemma}
In fact, \eqref{eq:lambdarep} holds for any $X\in\X$  and any decreasing $\Lambda: \R\to[0,1]$  that is not constantly $0$.
{ \begin{remark}
Note that Lemma \ref{lem:representation} does not hold for a general $\Lambda$ even when $\Lambda$ is increasing.  For example, consider $\Lambda(x) = {(1 + x)}/{4}$ for $0 \leq x < 1$, and $\Lambda(x) = \frac{1}{2}$ for $x \geq 1$.  Let  X be a random variable with  $X\sim U(0,1)$. We can compute $\Lambda\VaR(X) = 1/3>\inf_{x \in \mathbb{R}} \left\{ \operatorname{VaR}_{\Lambda(x)}(X) \vee x \right\} = 1/4$, showing that the infimum representation does not hold. 
\end{remark}}
Using Lemma \ref{lem:representation}, we can transform the optimization problem in \eqref{eq:obj_general} into an optimization problem based on VaR.  

\begin{proposition}\label{prop:var_problem}
    For $X\in\X_+$ and $\Lambda\in\mathcal D$, we have
    \begin{equation}\label{eq:prop1}\inf_{f\in\F}\Lambda\VaR(T_f(X))=\inf_{x\in\R_+}\left\{\inf_{f\in\F}\left\{\VaR_{\Lambda(x)}(X)-\VaR_{\Lambda(x)}(f(X))+\Pi(f(X))\right\}\vee x\right\}.\end{equation}
\end{proposition}
\begin{proof}
    By Lemma \ref{lem:representation}, we have
    $$\begin{aligned}
       \inf_{f\in\F}\Lambda\VaR(T_f(X))&=\inf_{f\in\F}\inf_{x\in\R_+}\left\{\VaR_{\Lambda(x)}(T_f(X))\vee x\right\}\\
        &=\inf_{x\in\R_+}\inf_{f\in\F}\left\{\VaR_{\Lambda(x)}(T_f(X))\vee x\right\}\\
        &=\inf_{x\in\R_+}\left\{\inf_{f\in\F}\left\{\VaR_{\Lambda(x)}(X)-\VaR_{\Lambda(x)}(f(X))+\Pi(f(X))\right\}\vee x\right\},
    \end{aligned}$$
    where the last equality is because $\VaR$ is cash additive and  comonotonic additive,  and the term $x$ is independent of $f\in\F$. {The exchange of the two infima is valid because the minimization over $f$ and $x$ are independent.} The proof is complete.
\end{proof}
{ Proposition \ref{prop:var_problem} shows that, under certain assumptions on $\Lambda$, the optimization problem in \eqref{eq:obj_general} can be decomposed into two steps: \begin{itemize}\item[(i)] For a fixed $x$ (i.e., a fixed confidence level $\Lambda(x)$), we solve for the optimal ceded insurance contract $f(X)=f(X, \Lambda(x))$ that minimizes $\VaR_{\Lambda(x)} \left( T_f(X) \right)$. \item[(ii)] We  determine the optimal $x^*$ by solving 
$ \inf_{x \in \mathbb{R}_+} \left\{ \VaR_{\Lambda(x)} \left( T_{f(X, \Lambda(x))} (X) \right) \vee x \right\}.
$ The optimal ceded function is then given by $f^*(X) = f(X, \Lambda(x^*))$.
\end{itemize}

}
\section{Optimal results under expected value premium principle}\label{sec:EVpremium}
In this section, we assume that the insurer prices the indemnity functions using the expected value premium principle, a method that is widely adopted in insurance literature due to its simplicity and economic relevance:

\begin{equation}\label{eq:ev_p}
\Pi(f(X)) = (1 + \theta) \E[f(X)],
\end{equation}
where $\theta \ge 0$ represents the safety loading parameter. { Here, this premium principle can be seen as the indifference premium of a risk-neutral insurer, and $\theta$ reflects any additional costs. 
That is, it is a solution of the following equation:
\begin{equation}\label{eq:expected}
\E[-W+f(X)-\Pi(f(X))+C(f(X)]=\E[-W],
\end{equation}
where $W\in\R$ is the initial wealth, and $C(f(X))=\theta \E[f(X)]$ represents the costs of the insurer, assumed to be linear in the expected indemnity.}

\subsection{Optimal retention with deductible contract}\label{sec:3.2}
We first focus on {a specific but rather important class of insurance contracts in $\F$, the stop-loss insurance contracts. Within this class,} the insurer covers losses on the DM above a  deductible $l\in[0,\infty]$, that is,  $f(x)=(x-l)_+$ for $x\in\R_+$.  {We are particularly interested in stop-loss contracts because they are widely applied in practice, especially in the property and casualty insurance sectors. Moreover, they have been extensively studied in the literature;} see, e.g. \cite{B60},  \cite{CSYY14}, \cite{KM20}, \cite{LM22}, and \cite{CLY24}.  
 The risk exposure of the DM after purchasing insurance is given by
 \begin{equation}\label{eq:T_l}T_l(X)= X\wedge l+(1+\theta)\E[(X-l)_+].\end{equation}  Next, we  address the following optimization problem:
\begin{equation}\label{eq:obj_deductible}
    \inf_{l\in\R_+}\Lambda\VaR(T_l(X))
\end{equation}
by sequentially solving the two optimization problems in \eqref{eq:prop1}; {note that Proposition \ref{prop:var_problem} also holds if $\F$ is changed to a different admissible set of indemnities.}
 {In particular, the inner optimization problem is addressed by Theorem 2.1 in \cite{CT07}. However, as we relax the constraints on the loss distribution, specifically the assumption that the distribution function is continuous and strictly increasing, a modified proof is provided to accommodate this relaxation.}
 
Let $d^*=\VaR_{\theta^*}(X)$, where $\theta^*=\theta/(1+\theta)$. It is clear that $d^*=0$ for $\theta=0$.    Write $$M=d^*+(1+\theta)\E[(X-d^*)_+].$$
\begin{theorem}\label{thm:deductbile} Let $\Lambda\in\mathcal D$.  The solution $l^*\in[0,\infty]$ that solves the $\Lambda\VaR$-based insurance model   \eqref{eq:obj_deductible}   is  
given by 
\begin{equation}\label{eq:dec}
l^*= \begin{cases}d^* , &\text{if}~ M< \Lambda\VaR(X),  \\ \infty, & \text { otherwise }, \end{cases}
\end{equation}
and the corresponding minimum is given by $$\Lambda\VaR\left(T_{l^*}(X)\right)=M \wedge \Lambda\VaR(X).$$

\end{theorem}

 \begin{proof}
 {If $\E[X]=\infty$, then $M=\infty$. As $\p(X<\infty)=1$,  we have $\Lambda\VaR(X)<\Lambda\VaR(T_l(X))=\infty$ for all $l\in[0,\infty)$. Therefore, $l^*=\infty$ and $\Lambda\VaR\left(T_{l^*}(X)\right)= \Lambda\VaR(X)<\infty=M.$ For the rest of the proof, we assume that $\E[X]<\infty$.}
  By Proposition \ref{prop:var_problem}, we have
    \begin{equation}\label{eq:var_deductible}
        \inf_{l\in\R_+}\Lambda\VaR(T_l(X))=\inf_{x\in\R_+}\left\{\inf_{l\in\R_+}\left\{\VaR_{\Lambda(x)}(X\wedge l)+(1+\theta)\E[(X-l)_+]\right\}\vee x\right\}.
    \end{equation}
{  
For a fixed $x\in\R_+$,  observe  from \eqref{eq:var_deductible} that 
    $$g(l):=\VaR_{\Lambda(x)}(X\wedge l)+(1+\theta)\E[(X-l)_+] $$
      is  decreasing  for  $l \in\left[ \VaR_{\Lambda(x)}(X), \infty\right)$,  with the limit $\VaR_{\Lambda(x)}(X)$ as $l \rightarrow \infty$.  For    $l \in\left[0, \VaR_{\Lambda(x)}(X)\right)$, $g(l)$   decreases  for $l\in [0, \VaR_{\Lambda(x)}(X) \wedge d^*)$ and increases for $l\in[ \VaR_{\Lambda(x)}(X)\wedge d^*, \VaR_{\Lambda(x)}(X))$. 
Therefore,  the function $g(l)$ attains its minimum value   at  either $d^*\geq 0$ or as $l\to\infty$. }   
    
When $\VaR_{\Lambda(x)}(X)>M$, we have
    $$\inf_{l\in\R_+}\left\{\VaR_{\Lambda(x)}(X\wedge l)+(1+\theta)\E[(X-l)_+]\right\}=M,$$
    with the corresponding optimal retention given by $l^*=d^*$;
    when $\VaR_{\Lambda(x)}(X)\leq M$, we have
    $$\inf_{l\in\R_+}\left\{\VaR_{\Lambda(x)}(X\wedge l)+(1+\theta)\E[(X-l)_+]\right\}=\VaR_{\Lambda(x)}(X),$$
    with the corresponding optimal retention given by $l^*=\infty$. For $x\in\R_+$, define 
    $$M(x)=M\id_{\{\VaR_{\Lambda(x)}(X)>  M\}}+\VaR_{\Lambda(x)}(X)\id_{\{\VaR_{\Lambda(x)}(X)\leq   M\}}.$$
    Hence, \eqref{eq:var_deductible} becomes
$$\inf_{l>0}\Lambda\VaR(T_l(X))=\inf_{x\in\R_+}\left\{M(x)\vee x\right\}.$$
Let 
$$\bar x=\sup\{x\in\R_+:\VaR_{\Lambda(x)}(X)>M\}\mbox{~~and~~}x^*=\inf\{x\in\R_+:M(x)\le x\}.$$
Here, $\bar x$ can take the value $\infty$,  but we have $x^*<\infty$.   Therefore, 
if $x^*> \bar x$, as $\Lambda$ is decreasing, we have $$\inf_{l>0}\Lambda\VaR(T_l(X))=\inf_{x\in\R_+}\left\{M(x)\vee x\right\}=\inf_{x\in\R_+}\left\{\VaR_{\Lambda(x)}(X)\vee x\right\}=\Lambda \VaR(X).$$
Moreover, we have 
$$\inf_{l\in\R_+}\Lambda\VaR(T_l(X))=\inf_{x\in\R_+}\left\{M(x)\vee x\right\}\le \inf_{x\in\R_+}\left\{M\vee x\right\}=M.$$
Thus, if $x^*>\bar x$, $\inf_{l\in\R_+}\Lambda\VaR(T_l(X))=\Lambda \VaR(X)\le M$.
If $x^*\leq \bar x$, $\inf_{l\in\R_+}\Lambda\VaR(T_l(X))=M< \Lambda \VaR(X).$ The inequality holds as  
$$\Lambda\VaR(X)=\inf_{x\in\R_+}\left\{\VaR_{\Lambda(x)}(X)\vee x\right\}> \inf_{x\in\R_+}\left\{M(x)\vee x\right\}=\inf_{l\in\R_+}\Lambda\VaR(T_l(X)).$$ Thus we have the desired result.
 \end{proof}

\begin{remark}
When a quota-share insurance contract is employed, the optimal control takes a `bang-bang' type, with the values determined by the relationship between \(\VaR_\alpha(X)\) and \((1 + \theta) \mathbb{E}[X]\). This characteristic arises due to the positive homogeneity of VaR. Given that the problem of minimizing \(\Lambda\VaR\) can be transformed into a VaR minimization problem, following the steps similar to those in the proof of Theorem \ref{thm:deductbile},  it is not difficult to verify  that the optimal solution for \(\Lambda\VaR\) under the quota-share insurance contract also retains the `bang-bang'  form.
\end{remark}

We can see from Theorem \ref{thm:deductbile} that the deductible in \eqref{eq:dec} coincides with that for minimizing \(\VaR_\alpha(T_l)\) with \(\alpha \in (0,1)\) when \(M \leq \min(\VaR_\alpha(X), \Lambda \VaR(X))\).
When \( l^* = \infty \),  it indicates that  the DM buys no insurance. The formulation of \(l^*\) depends solely on the safety loading factor and remains entirely unaffected by the functional form of \(\Lambda\). 
{Moreover,  we   have  $f^*(x)=x$ when $\theta=0$,  which implies that   the DM chooses to transfer all the risk to  the insurer when $\mathbb E[X]\leq \Lambda\VaR(X)$ (see Equation \eqref{eq:dec}) as the premium represents a ``cheap" insurance.}

 In the next result, we provide a necessary and sufficient condition for the existence of a positive and finite deductible, indicating that the DM always chooses to transfer part of the claim to the insurer.
\begin{theorem} \label{prop:2} {Assume that $\E[X]<\infty$ and  $\Lambda\in\mathcal D$.} The solution $0<l^*<\infty$ for problem \eqref{eq:obj_deductible} exists if and only if $\theta^*>F_X(0)$ and $\Lambda(M-\epsilon)> F_X(M-\epsilon)$ for all $\epsilon\in(0,M]$.
\end{theorem}
\begin{proof} 

We first show the ``if" part. If $\theta^*>F_X(0)$, we have $\VaR_{\theta^*}(X)>0$.
By the proof of Theorem \ref{thm:deductbile}, we can write
$$\inf_{l\in\R_+}\Lambda\VaR(T_l(X))=\inf_{\VaR_{\Lambda(x)}(X)\ge M}\left\{M\vee x\right\}\wedge \inf_{\VaR_{\Lambda(x)}(X)< M}\left\{\VaR_{\Lambda(x)}(X)\vee x\right\}.$$
For $x\in\R_+$ such that $\VaR_{\Lambda(x)}(X)\ge M$, we have $\Lambda(x)\ge F_X(M)$. Because $\Lambda$ is decreasing, we have $x\in[0,a]$ for some $a\ge 0$ and thus 
$$\inf_{\VaR_{\Lambda(x)}(X)\ge M}\left\{M\vee x\right\}=M.$$
For $x\in\R_+$ such that $\VaR_{\Lambda(x)}(X)<M$, we have $\Lambda(x)\le F_X(M)$. Since $\Lambda(M-\epsilon)>F_X(M-\epsilon)$ for all $\epsilon\in(0,M]$ and $F_X$ is continuous, we have $\Lambda(M-)\ge F_X(M)\ge \Lambda(x)$.
    When $\Lambda(M-)>F_X(M)$, we have $\Lambda(M-)>\Lambda(x)$ and thus $x\ge M>\VaR_{\Lambda(x)}(X)$. It follows that
    \begin{equation}\label{eq:case1}
        \inf_{\VaR_{\Lambda(x)}(X)< M}\left\{\VaR_{\Lambda(x)}(X)\vee x\right\}=\inf_{\VaR_{\Lambda(x)}(X)< M}x\ge M.
    \end{equation}
    When $\Lambda(M-)=F_X(M)$, we have $\VaR_{\Lambda(M-)}(X)=M$. Thus $x<M\le\VaR_{\Lambda(x)}(X)$ for all $x< M$ and $x\ge\VaR_{\Lambda(x)}(X)$ for all $x\ge M$. It follows that 
    \begin{equation}\label{eq:case2}
        \inf_{\VaR_{\Lambda(x)}(X)< M}\left\{\VaR_{\Lambda(x)}(X)\vee x\right\}=\inf_{\VaR_{\Lambda(x)}(X)< M,~x\ge M}x\ge M.
    \end{equation}
    Therefore, we have by \eqref{eq:case1} and \eqref{eq:case2} that $\inf_{\VaR_{\Lambda(x)}(X)< M}\left\{\VaR_{\Lambda(x)}(X)\vee x\right\}\ge M$.

    According to the discussion above, we have
$$\inf_{l\in\R_+}\Lambda\VaR(T_l(X))=\inf_{\VaR_{\Lambda(x)}(X)\ge M}\left\{M\vee x\right\}=M,$$
    and the corresponding optimal retention is given by $l^*=\VaR_{\theta^*}(X)$.

    To prove the ``only if" part for the existence of optimal retention, we first suppose for contradiction that $\theta^*\le F_X(0)$. It follows that $l^*=\VaR_{\theta^*}(X)\le 0$. This leads to a contradiction with the fact that $l^*>0$. Next, we suppose for contradiction that $\Lambda(b)\le F_X(b)$ for some $b<M$. It follows that $\VaR_{\Lambda(b)}(X)\le b$, and thus
    \begin{equation*}
        \inf_{\VaR_{\Lambda(x)}(X)< M}\left\{\VaR_{\Lambda(x)}(X)\vee x\right\}=\inf_{\VaR_{\Lambda(x)}(X)< M,~x<b}\left\{\VaR_{\Lambda(x)}(X)\vee x\right\}\wedge b\le b<M.
    \end{equation*}
    It yields that
    $$\inf_{l\in\R_+}\Lambda\VaR(T_l(X))<M$$
    with corresponding optimal retention $l^*=\infty$. This contradicts with the fact that $l^*<\infty$.

    Therefore, if the solution $0<l^*<\infty$ for \eqref{eq:obj_deductible} exists, then $\theta^*>F_X(0)$ and $\Lambda(M-\epsilon)\ge F_X(M-\epsilon)$ for all $\epsilon\in(0,M]$. This completes the proof.
\end{proof}


{ Although the inner problem in \eqref{eq:obj_deductible} and the optimal deductible in Theorem \ref{thm:deductbile} are similar to the results for VaR-based optimal insurance problems established in \cite{CT07} and \cite{CT11},} Theorem \ref{prop:2} provides a non-trivial characterization result for the existence of a positive optimal deductible, which is clearly more general than that for the $\VaR$ problem. Specifically, if \(\Lambda(x) \equiv \alpha\), the second condition in Theorem \ref{prop:2} simplifies to \(F(M) < \alpha\), implying \(M < \VaR_\alpha(X)\).  This directly aligns with the result in Theorem 2.1 of \cite{CT07}, which is a corollary of Theorems \ref{thm:deductbile} and \ref{prop:2}.
{
\begin{corollary} If $\Lambda(x) \equiv \alpha$ for any $x \ge 0$, where $\alpha \in (0,1)$,   the optimal  deductible $l^*>0$ for the problem \eqref{eq:obj_deductible} exists if and only if both conditions $\theta^*>F_X(0)$
and
$
M<\VaR_\alpha(X)
$ are 
satisfied. Furthermore, if $l^*$   exists, then it holds that 
$
l^ * = d^ *.
$
\end{corollary}}

\subsection{Optimal retention with  general contract }
{In this section, we explicitly derive the optimal solution for the  ceded loss functions in the general form defined by \eqref{eq:marginal}. Under the expected value premium principle, we aim to solve the following problem:
\begin{equation}\label{eq:general}
    \inf_{f\in\mathcal F}\Lambda\VaR(X-f(X)+(1+\theta)\E[f(X)]).
\end{equation}
}
For  $x\in\R_+$,  denoted by  \begin{equation}\label{eq:G}G(x)= d^* \wedge \VaR_{\Lambda(x)}(X)   +(1+\theta) \mathbb{E}\left[\min \left\{\left(X-d^*\right)_{+},\left(\VaR_{\Lambda(x)}(X)-d^*\right)_{+}\right\}\right].\end{equation}

	\begin{theorem}\label{thm:1} Let $\Lambda\in\mathcal D$, and {suppose that  $\Lambda$ is right-continuous.} The optimal $f^{*}$ that solves the $\Lambda\VaR$-based insurance model \eqref{eq:general} over the class of ceded loss functions $\F$ is given by
\begin{equation}\label{eq:f_general}
f^{*}(x)= \begin{cases}\min \left\{\left(x-d^*\right)_{+}, \VaR_{\Lambda(x^*)}(X)-d^*\right\}, &\text{if}~d^*< \VaR_{\Lambda(x^*)}(X),  \\ 0, & \text { otherwise },\end{cases}
\end{equation}
in which  \begin{equation}\label{eq:x_star}x^*=\inf\{x\in\R_+: G(x)\leq x\}<\infty.\end{equation}
Moreover, 
$\Lambda\VaR\left(T_{f^*}(X)\right)=x^*.$
\end{theorem}
	\begin{proof} (i) Given $x\in\R_+$, we first solve the inner minimization of \eqref{eq:prop1}, that is, {
$\inf_{f\in\F} g(x)$, where $$g(x): = \VaR_{\Lambda(x)}(X)-\VaR_{\Lambda(x)}(f(X))+\Pi(f(X)).$$}     We  can rewrite $g(x)$ as  $$\begin{aligned}g(x)= \int_0^\infty h(S_X(y))\d y + \int_0^\infty ((1+\theta)S_X(y) - h(S_X(y)))q(y)\d y, 
 \end{aligned}$$
where 
$h(t)=\id_{\{t>1-\Lambda(x)\}},$  and it is clear that the following function $q$ minimizes $g$: 
\begin{equation}\label{eq:q1}q(y)=\left\{\begin{aligned}&0,~~~~~&\text{if }(1+\theta)S_X(y) - h(S_X(y))>0,\\&1,~~~~~&\text{if }(1+\theta)S_X(y) - h(S_X(y))<0,\\& c(y),~~~&\text{otherwise},\end{aligned}\right.\end{equation}
where $c$ could be any $[0,1]$-valued {Lebesgue-measurable} function on $\{y:(1+\theta)S_X(y) - h(S_X(y))=0\}$.  Define $$H(y)=(1+\theta)S_X(y) - h(S_X(y)).$$

{
To give an explicit form of $q$, it suffices to study the sign of $H$. Noting that $S_X(y)> 1-\Lambda(x)$ if and only if $y<\VaR_{\Lambda(x)}(X)$, we have 
\begin{equation*}H(y)=\left\{\begin{aligned}&(1+\theta)S_X(y)-1,~~~~~&\text{if }y<\VaR_{\Lambda(x)}(X),\\&(1+\theta)S_X(y),~~~~~&\text{if }y\ge\VaR_{\Lambda(x)}(X).\end{aligned}\right.\end{equation*}
 Recall that $d^*=\VaR_{\theta^*}(X)$ where $\theta^*=\theta/(1+\theta)$.
When $y<\VaR_{\Lambda(x)}(X)$, $H(y)>0$ if $y\in[0,\VaR_{\Lambda(x)}(X)\wedge d^*)$ and $H(y)\le 0$ if $y\in[\VaR_{\Lambda(x)}(X)\wedge d^*,\VaR_{\Lambda(x)}(X))$. When $y\ge \VaR_{\Lambda(x)}(X)$, $H(y)\ge 0$. Thus,  we can set  $q(y)=0$ for $y\in[0,d^* \wedge \VaR_{\Lambda(x)}(X))\cup[\VaR_{\Lambda(x)}(X),\infty),$
 and $ q(y) =  1$ for  $y\in[\VaR_{\Lambda(x)}(X)\wedge d^*,\VaR_{\Lambda(x)}(X))$, which implies $f^*(y)=\min \left\{\left(y-d^*\right)_{+}, (\VaR_{\Lambda(x)}(X)-d^*)_+\right\}$.
 Therefore,  we can obtain that $ \inf_{f\in \F} g(x)=G(x).$
}

 

(ii) We are now ready to solve the outer optimization problem, that is,   $$\inf_{f\in \F}\Lambda\VaR(T_f(X))=\inf_{x\in\R_+}\left\{G(x)\vee x\right\}.$$  
It is evident that      $G(x)$  decreases in $x$ as   $\VaR_{\Lambda(x)}(X) $ decreases in $x$. 
Define $$\bar x=\sup\{x\in\R_+: \VaR_{\Lambda(x)}(X)>d^*\}, ~~\text{and}~~x^*=\inf\{x\in\R_+: G(x)\leq x\}.$$
{Note that $\bar x$ can be $\infty$.  Further, as $\Lambda(0)<1$ and $\p(X<\infty)=1$, we have $x^*<\infty$.} Let $$G_1(x)= d^*   +(1+\theta) \mathbb{E}\left[\min \left\{\left(X-d^*\right)_{+},(\VaR_{\Lambda(x)}(X)-d^*)_+\right\}\right].$$  
{Since $\Lambda(x)$ is right-continuous in $x$, we have $G(x)=G_1(x)$ for $x<\overline x$ and $G(x)=\VaR_{\Lambda(x)}$ for $x\geq \overline x$.} Additionally, if   $x<\overline{x}$,  we have  $d^*<\VaR_{\Lambda(x)}(X)$, then $$\begin{aligned}G(x)&=d^*   +(1+\theta) \mathbb{E}\left[\min \left\{\left(X-d^*\right)_{+},(\VaR_{\Lambda(x)}(X)-d^*)_+\right\}\right]\\&< d^*+(1+\theta)(\VaR_{\Lambda(x)}(X)-d^*)(1-F(d^*))\leq \VaR_{\Lambda(x)}(X).\end{aligned}$$ The second inequality is because $F(d^*)=\theta^*={\theta}/{(1+\theta)}$. Thus,   if $x^*< \overline{x}$,   $\inf_{x\in\R_+}\left\{G(x)\vee x\right\}=x^*<\Lambda\VaR(X) $, and if $x^*\geq\overline{x}$, $\inf_{x\in\R_+}\left\{G(x)\vee x\right\}=x^*=\Lambda\VaR(X)$.  The relationship between $G(x)$ and $x$ in the proof above is illustrated clearly in Figure \ref{fig1}. 
{ As $\Lambda$ is right-continuous, the infimum is attainable, thus $\min_{x\in\R_+}\{G(x)\vee x\}=x^*$ and the minimum is attained at $x=x^*$.}
Therefore, { according to (i)}, we have $\inf_{f\in \F}\Lambda\VaR(T_f(X))=x^*.$
\begin{figure}[htb!]
\centering \includegraphics[width=15cm]{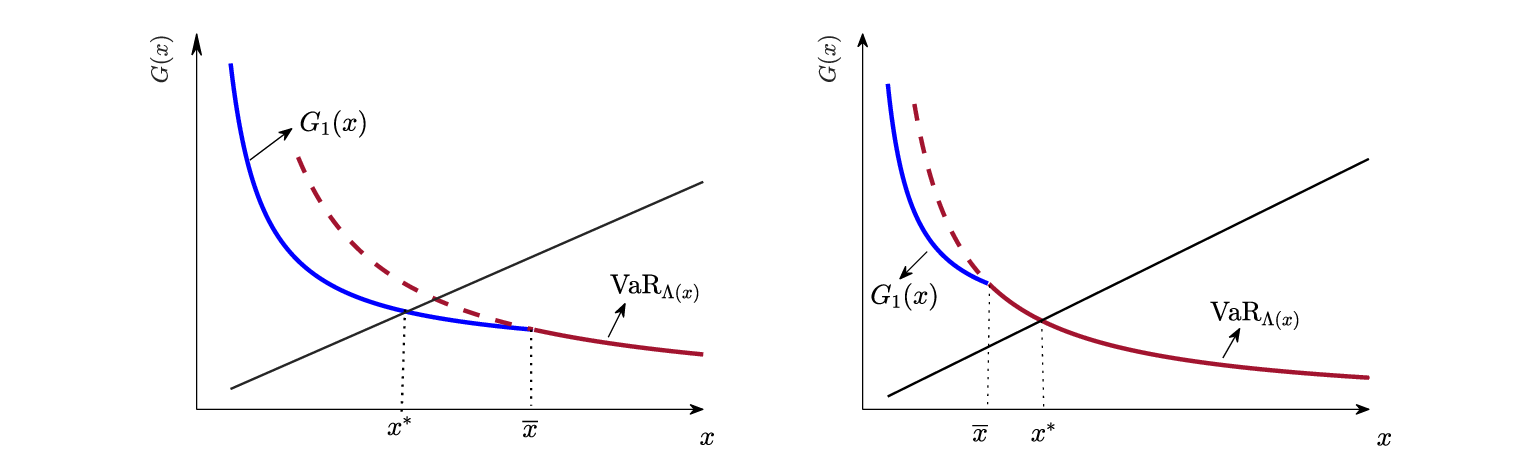}
 \captionsetup{font=small}
\caption{An illustration of the relationship between $G(x)$ and $x$}\label{fig1}
\end{figure}
\end{proof}

{\begin{remark} Theorem 3.2 of \cite{CT11}   provides the optimal ceded loss function for the case of VaR. Thus, for the first part of Theorem \ref{thm:1}, we can directly apply the result from  \cite{CT11} to the inner minimization problem and solve the outer optimization problem accordingly. However, in this paper, we use a different method based on the marginal indemnity function, and therefore, we provide a separate proof. 
 
Additionally, by comparing \(x^*\) in Theorem \ref{thm:1} with \(M\) in Theorem \ref{thm:deductbile},  and noting that   \(M \geq G(x)\) for any \(x \in \mathbb{R}_+\), we obtain $$ \Lambda\VaR\left(T_{l^*}(X)\right)\geq \Lambda\VaR\left(T_{f^*}(X)\right).$$ This result is  expected, as  the admissible set of 
 $f$ includes the special form of stop-loss insurance contracts.  In contrast,  for a general ceded loss function, the optimal form is a limited stop-loss with an upper limit of
 $\VaR_{\Lambda(x^*)}(X)-d^*$.  Finally, if $\theta=0,$ meaning that  the premium corresponds to a cheap insurance, the DM always  transfers all the risk to the insurer when $x< \VaR_{\Lambda(x^*)}(X)$. \end{remark}}
 
In the following corollary, we present a special case where $\Lambda\VaR$ simplifies to $\VaR$.
	\begin{corollary}\label{cor:1}
If $\Lambda(x) \equiv \alpha$ for any $x \ge 0$, where $\alpha \in (0,1)$,  the solution to \eqref{eq:obj_general} is given by   $$
f^{*}(x)= \begin{cases}\min \left\{\left(x-d^*\right)_{+}, \VaR_{\alpha}(X)-d^*\right\}, &\text{if}~d^*\le \VaR_{\alpha}(X),  \\ 0, & \text { otherwise },\end{cases}
$$ and thus $$\VaR_\alpha\left(T_{f^*}(X)\right)=d^* \wedge \VaR_{\alpha}(X)   +(1+\theta) \mathbb{E}\left[\min \left\{\left(X-d^*\right)_{+},\left(\VaR_{\alpha}(X)-d^*\right)_{+}\right\}\right].$$ 
This is consistent with the optimal indemnity function obtained by \cite{CT11}.
\end{corollary} 
 
By comparing Theorem \ref{thm:1} and Corollary \ref{cor:1},    we observe that the optimal ceded loss function for  $\Lambda\VaR$ is identical to the one for  VaR  when $d^* \leq \min(\VaR_\alpha(X), \VaR_{\Lambda(x^*)}(X))$. 
 However,   $\Lambda\VaR$ can offer  greater flexibility compared to a VaR-based contract  in determining the upper limit of the ceded loss function, as  $x^*$ depends on $\Lambda$ and $\theta$. 

\section{Optimal results under $\mathbf{\Lambda}$VaR premium principle}\label{sec:LambdaVaR}
{This section studies the case in which a Lambda-VaR risk measure is incorporated into the premium calculation. Using a $\Lambda\VaR$ can be attractive for the insurer for the following two reasons. First, the VaR is often used as a quantile or percentile premium \citet[e.g., Section 5 in ][]{KGDD08}, where the maximum probability of a loss in a contract is bounded by a fixed value. $\Lambda\VaR$ can be seen as an extension of the quantile premium, where the quantile parameter is not fixed but depends on the loss distribution itself. This setting allows $\Lambda\VaR$ to flexibly adjust the premium based on the size of the loss and the risk preferences, making it more adaptable for different risk scenarios. 
Second, suppose that the insurer measures the risk position after selling the insurance indemnity $f(X)$ for the deterministic premium $\Pi\geq0$ by: $\Lambda'\VaR(f(X)-\Pi)$, where $\Lambda'\in\mathcal D$.\footnote{{Here, we  extend the domain of $\Lambda'\VaR$ to the set $\mathcal X$, where $\Lambda$ is a decreasing function defined on $\mathbb{R}$.}} By cash subadditivity of $\Lambda'\VaR$ \citep[see][]{HWWX25}, a premium of the insurer that satisfies \begin{equation*}
\Lambda'\VaR(f(X)-\Pi)\leq\Lambda'\VaR(0)=0
\end{equation*}
must satisfy $\Pi\geq\Lambda'\VaR(f(X))$.
Different from \eqref{eq:expected} for the expected premium principle, using the $\Lambda'\VaR$ premium principle mitigates the total risk position of the insurer compared with no insurance when he measures the risk by $\Lambda'\VaR$. This may be preferable for risk-averse insurers in some situations.
Therefore, we define the premium principle for a general 
$f\in\mathcal F$ as 
\begin{equation*} \Pi(f(X))=\Lambda'\VaR(f(X)),\end{equation*} and aim to solve the following optimization problem
\begin{equation}\label{eq:ev}
    \inf_{f\in\F}\Lambda\VaR(T_f(X))= \inf_{f\in\F}\Lambda\VaR(X-f(X)+\Lambda'\VaR(f(X)) )
\end{equation}
with   $\Lambda, \Lambda'\in\mathcal {D}$.}

It is easy to verify that for $\alpha, \beta\in(0,1)$, $f^*(x)=x\id_{\{\alpha\geq \beta\}}$ is a  solution to the following optimization problem 
\begin{equation}\label{eq:ev0}
    \inf_{f\in\F}\VaR_\alpha(X-f(X)+\VaR_\beta(f(X)).
\end{equation}
A natural question arises: Whether a similar observation can still be made  for the optimization problem \eqref{eq:ev}.  
Note that $\Lambda\VaR$ is monotonic with respect to $\Lambda$ in the sense that  if $\Lambda_1 \leq \Lambda_2$ then each of the $\Lambda_1\VaR$ is smaller than the corresponding $\Lambda_2\VaR$; see Proposition 3c of \cite{BP22}.  

 In the following proposition, we show that for the $\Lambda\VaR$ premium principle, a bang-bang strategy can still be optimal. 
\begin{proposition}\label{prop:3}For  $\Lambda, \Lambda'\in\mathcal {D}$, the optimal $f^{*}$ that solves  the $\Lambda\VaR$-based insurance model \eqref{eq:ev} over the class of ceded loss functions $\F$ is given by
\begin{equation}\label{eq:bang_bang}
f^{*}(x)= \begin{cases} x, &\text{if}~ \Lambda'\VaR(X)   \leq \Lambda\VaR(X) ,   \\0 , &\text{if}~ \Lambda'\VaR(X)   >\Lambda\VaR(X).\end{cases}
\end{equation} 
The corresponding minimum is given by $$\Lambda\VaR\left(T_{f^*}(X)\right)=\Lambda\VaR(X)  \wedge \Lambda'\VaR(X).$$   
\end{proposition}
\begin{proof}Note that, by  Lemma \ref{lem:representation} and Proposition \ref{prop:var_problem},
$$\begin{aligned}\inf_{f\in\F}\Lambda\VaR(T_f(X))&=\inf_{x\in\R_+}\left\{\inf_{f\in\F}\left\{\VaR_{\Lambda(x)}(X)-\VaR_{\Lambda(x)}(f(X))+\Pi(f(X))\right\}\vee x\right\}\\&  =\inf_{x\in\R_+}  \left\{\inf_{f\in\F}\left \{ \VaR_{\Lambda(x)}(X)-\VaR_{\Lambda(x)}(f(X))+ \inf_{y \in \R_+} \left\{\VaR_{\Lambda'(y)} (f(X)) \vee y \right\} \right \} \vee x \right\}\\&= \inf_{x\in\R_+}  \left\{\inf_{y \in \R_+} \inf_{f\in\F}    \left \{ \VaR_{\Lambda(x)}(X)-\VaR_{\Lambda(x)}(f(X))+ \VaR_{\Lambda'(y)} (f(X)) \vee y   \right \} \vee x \right\} .\end{aligned}$$   

 Fix $x,y\ge 0$. If $\Lambda(x)\leq \Lambda'(y)$, we have $f^*(X)=0$.
 This is because for any $f(X)>0$ and $y$ such that $\Lambda(x)\leq \Lambda'(y)$, we have 
 $$\VaR_{\Lambda(x)}(X)-\VaR_{\Lambda(x)}(f(X))+ \VaR_{\Lambda'(y)} (f(X)) \vee y \ge \VaR_{\Lambda(x)}(X),$$ which implies that 
 $$\begin{aligned}&\inf_{x\in\R_+}  \left\{\inf_{y \in \R_+}   \left \{ \VaR_{\Lambda(x)}(X)-\VaR_{\Lambda(x)}(f(X))+ \VaR_{\Lambda'(y)} (f(X)) \vee y   \right \} \vee x \right\}\\&\geq  \inf_{x\in\R_+}  \left\{ \VaR_{\Lambda(x)}(X)   \vee x \right\}= \Lambda\VaR(X).\end{aligned} $$
 On the other hand, if   $\Lambda(x)>\Lambda'(y)$, we   have $f^*(X)=X$.
 Indeed, for any $f\in\mathcal{F}$, if $\VaR_{\Lambda'(y)} (f(X))\le y$, then
$$\begin{aligned}
    &\VaR_{\Lambda(x)}(X)-\VaR_{\Lambda(x)}(f(X))+ \VaR_{\Lambda'(y)} (f(X)) \vee y\\&\ge \VaR_{\Lambda'(y)}(X)-\VaR_{\Lambda'(y)}(f(X))+y\ge \VaR_{\Lambda'(y)}(X),
\end{aligned}$$
while if  $\VaR_{\Lambda'(y)} (f(X))> y$, then
$$\begin{aligned}
    &\VaR_{\Lambda(x)}(X)-\VaR_{\Lambda(x)}(f(X))+ \VaR_{\Lambda'(y)} (f(X)) \vee y\\&\ge  \VaR_{\Lambda'(y)}(X)-\VaR_{\Lambda'(y)}(f(X))+ \VaR_{\Lambda'(y)} (f(X))= \VaR_{\Lambda'(y)} (X).
\end{aligned}$$
Thus, we can see that the optimal contract is either full or no insurance, and its form does not depend on \(x\) and \(y\).  Therefore, if $f^*(X)=X$, we have 
$$ \inf_{f\in\F}\Lambda\VaR(T_f(X))=\inf_{x\in\R_+}  \left\{\inf_{y\in\R_+}     \left \{ \ \VaR_{\Lambda'(y)} (X) \vee y   \right \} \vee x \right\} =\Lambda'\VaR(X).$$
Otherwise, we have 
$$ \inf_{f\in\F}\Lambda\VaR(T_f(X))=\inf_{x\in\R_+}  \left\{\inf_{y\in\R_+}   \left \{ \ \VaR_{\Lambda(x)} (X) \vee y   \right \} \vee x \right\} =\Lambda\VaR(X),$$ which completes the proof.
\end{proof}

\begin{remark} 
Note that the problem \eqref{eq:ev0} can also be interpreted as a risk-sharing problem between two agents, with VaR serving as their preference measure due to its cash additivity property. Therefore, having one agent bear all the risk can be an optimal solution.   
By Theorem 11 of \cite{XH24}, the solution to \eqref{eq:ev} also applies to the problem of risk sharing between two $\Lambda\VaR$ agents, even though $\Lambda\VaR$ is not cash additive. 
In fact, Theorem 11 of \cite{XH24} shows that for \(\Lambda_i: \mathbb{R}_{+} \rightarrow [0, 1]\) being decreasing for \(i \in \{1, \ldots, n\}\) with \(n \geq 2\) and \(\Lambda = \min_{1 \leq i \leq m} \Lambda_i\), we have
$$\inf_{(X_1, \ldots, X_m) \in \mathbb{A}_n(X)} \sum_{i=1}^m \Lambda_i\VaR(X_i) = \Lambda\VaR(X),$$
where
$$\mathbb{A}_n(X) = \{(X_1, \ldots, X_n) : X_i \in \X_+, \sum_{i=1}^n X_i = X, \text{ and } X_1, \ldots, X_n \text{ are comonotonic}\}.$$
For a comprehensive  discussion on VaR-based and $\Lambda\VaR$-based risk-sharing problems, see \cite{ELW18}, \cite{L24} and \cite{XH24}.

  \end{remark} 

The optimal ceded loss function  in Proposition \ref{prop:3}  is generally not unique, similar to the case of VaR. 
Although a bang-bang strategy can be mathematically optimal, it is often less desirable and practically applicable in real-world settings.   To address this, we modify the premium in the subsequent context as follows: $$\Pi(f(X))=\E[f(X)]+\theta(\Lambda'\VaR(f(X))-\E[f(X)])$$ where $\Lambda'\in \mathcal D$ and $\theta\in[0,1]$ represents a safety loading. 
{ Similarly to \eqref{eq:expected}, such a premium principle satisfies the following condition for a risk-neutral insurer:
$$\E[-W+f(X)-\Pi(f(X))+C'(f(X))]=\E[-W],$$
with the cost of the insurer $C'(f(X))=\theta(\Lambda'\VaR(f(X))-\E[f(X)])$.  This formulation accounts for both the expected value of the function $f(X)$ and an additional safety margin that scales with the deviation of $\Lambda' \VaR(f(X))$ from $\E[f(X)]$.}
The optimization problem in \eqref{eq:ev} becomes 

\begin{equation}\label{eq:obj_3}
    \inf_{f\in\F}\Lambda\VaR(X-f(X)+\E[f(X)]+\theta(\Lambda'\VaR(f(X))-\E[f(X)])).
\end{equation}
\begin{theorem}\label{thm:3} Let $\Lambda,{\Lambda'}\in\mathcal D$, and {suppose that  $\Lambda$ is right-continuous.} The optimal $f^{*}$ that solves the $\Lambda\VaR$-based insurance model  \eqref{eq:obj_3} over the class of ceded loss functions $\F$ is given by
\begin{equation}\label{eq:f4}
f^*(x)= x\wedge\VaR_{\Lambda(x^*)}(X),
\end{equation} in which \begin{equation}\label{eq:x4}x^*=\inf\{x\in\R_+: (1-\theta)\E[  X\wedge \VaR_{\Lambda(x)}(X)]+\theta (\Lambda'\VaR ( X) \wedge \VaR_{\Lambda(x)}(X))\leq x\}{<\infty}.\end{equation}  Moreover,  $\Lambda\VaR\left(T_{f^*}(X)\right)= x^*.$
	\end{theorem}

\begin{proof} By  Lemma \ref{lem:representation} and Proposition \ref{prop:var_problem}, we have
$$\begin{aligned}&\inf_{f\in\F}\Lambda\VaR(T_f(X))\\&=\inf_{x\in\R_+}\left\{\inf_{f\in\F}\left\{\VaR_{\Lambda(x)}(X)-\VaR_{\Lambda(x)}(f(X))+\Pi(f(X))\right\}\vee x\right\}\\&
=\inf_{x\in\R_+}\left\{\inf_{f\in\F}\left\{\VaR_{\Lambda(x)}(X)-\VaR_{\Lambda(x)}(f(X))+ \E[f(X)]+\theta(\Lambda'\VaR(f(X))-\E[f(X)])\right\}\vee x\right\}
\\&= \inf_{x\in\R_+}  \left\{\inf_{f\in\F}    \left \{ \VaR_{\Lambda(x)}(X)-\VaR_{\Lambda(x)}(f(X)) +(1-\theta)\E[f(X)]+\theta \inf_{y\in\R_+} \left\{ \VaR_{\Lambda'(y)} (f(X)) \vee   y  \right\}  \right \} \vee x \right\} \\&= \inf_{x\in\R_+}  \left\{\inf_{y\in\R_+} \inf_{f\in\F}    \left \{ \VaR_{\Lambda(x)}(X)-\VaR_{\Lambda(x)}(f(X)) +(1-\theta)\E[f(X)]+\theta \VaR_{\Lambda'(y)} (f(X)) \vee \theta  y   \right \} \vee x \right\} .\end{aligned}$$   
We first solve the inner optimization problem for a fixed  $x,y\ge 0$. Note that 
$$\begin{aligned}
 & \inf_{f\in\F}    \left \{ \VaR_{\Lambda(x)}(X)-\VaR_{\Lambda(x)}(f(X)) +(1-\theta)\E[f(X)]+\theta \VaR_{\Lambda'(y)} (f(X)) \vee \theta  y   \right \} \\ =&  \inf_{f\in\F,~y <\VaR_{\Lambda'(y)} (f(X))}   \left \{ \VaR_{\Lambda(x)}(X)-\VaR_{\Lambda(x)}(f(X)) +(1-\theta)\E[f(X)]+\theta \VaR_{\Lambda'(y)} (f(X))    \right \} \wedge\\&   \inf_{f\in\F,~y \ge\VaR_{\Lambda'(y)} (f(X))}   \left \{ \VaR_{\Lambda(x)}(X)-\VaR_{\Lambda(x)}(f(X)) +(1-\theta)\E[f(X)]+\theta  y   \right \}.
\end{aligned}$$
For a fixed $x,y\ge 0$, we solve the above two problems  with constraints separately: 

{\bf Step 1}:   \begin{equation}\label{eq:prob1}\begin{aligned} &\inf_{f\in\F, y<\VaR_{\Lambda'(y)} (f(X))  }   \left \{ \VaR_{\Lambda(x)}(X)-\VaR_{\Lambda(x)}(f(X)) +(1-\theta)\E[f(X)]+\theta \VaR_{\Lambda'(y)} (f(X))    \right \} \\& =\inf_{f\in\F,\lambda\in\R_+ }   \left \{ \VaR_{\Lambda(x)}(X)-\VaR_{\Lambda(x)}(f(X)) +(1-\theta)\E[f(X)]\right.\\
&\quad\left.+\theta \VaR_{\Lambda'(y)} (f(X))   +\lambda (y- \VaR_{\Lambda'(y)} (f(X))  ) \right \},\end{aligned}\end{equation}
where $\lambda\geq0$ is the Lagrange multiplier.  

For $x,y\ge 0$, define 
$$\begin{aligned}g_1(x,y)&= \int_0^\infty h_1(S_X(w))\d w \\
&+ \int_0^\infty ((1-\theta)S_X(w) - h_1(S_X(w))+(\theta-\lambda )  h_2(S_X(w)) )q(w)\d w+\lambda y, \end{aligned}$$ where $h_1(t)=\id_{\{t>1-\Lambda(x)\}}$   and $h_2(t)=\id_{\{t>1-\Lambda'(y)\}}.$ It is clear that the following $q_1$ will minimize $g_1$: 
\begin{equation*}q_1(w)=\left\{\begin{aligned}&0,~~~~~&\text{if }(1-\theta)S_X(w) - h_1(S_X(w))+(\theta -\lambda) h_2(S_X(w)) >0,\\&1,~~~~~&\text{if }(1-\theta)S_X(w) - h_1(S_X(w))+(\theta-\lambda)  h_2(S_X(w)) <0,\\& c(w),~~~&\text{otherwise},\end{aligned}\right.\end{equation*} where $c$ could be any $[0,1]$-valued {Lebesgue-measurable} function on $\{w:(1-\theta)S_X(w) - h_1(S_X(w))+(\theta-\lambda)  h_2(S_X(w)) =0\}$.  Let $$H_1(w)=(1-\theta)S_X(w) - h_1(S_X(w))+(\theta-\lambda)  h_2(S_X(w)) .$$
Note that for a fixed $x,y\ge 0$, we have two scenarios: $\Lambda(x)\leq\Lambda'(y)$ and  $\Lambda(x)>\Lambda'(y)$.
\\Case 1: \underline{$\Lambda(x)\leq\Lambda'(y)$}
\begin{itemize} \item If $S_X(w)> 1-\Lambda(x)$, that is, $w<\VaR_{\Lambda(x)}(X)\leq \VaR_{\Lambda'(y)}(X)$,  we  have $H_1(w)=(1-\theta)S_X(w)-1+(\theta-\lambda)<0.$  
\item If $S_X(w)\leq 1-\Lambda'(y)$, that is, $w\ge \VaR_{\Lambda'(y)}(X)$,  we have $H _1(w)=(1-\theta)S_X(w)\ge 0$.
\item If $1-\Lambda'(y)<S_X(w)\leq 1-\Lambda(x)$, that is, $\VaR_{\Lambda(x)}(X) \leq w< \VaR_{\Lambda'(y)}(X)$,  we have $H_1(w)=(1-\theta)S_X(w)+(\theta-\lambda) $. Let $d_0=\VaR_{\theta_0}(X)$ with $\theta_0=0\vee \frac{1-\lambda}{1-\theta}\wedge 1$. We have $H_1(w)\leq 0$ if $w\geq d_0$ and    $H_1(w)> 0$ if $w<d_0$. 
\end{itemize}	
    By the above analysis,   we can always construct a strategy  in the form of  $f(w)=w\wedge \VaR_{\Lambda(x)}(X) +(w- \tilde d )_+\wedge (\VaR_{\Lambda'(y)}(X) -\tilde d)$,
 where $\tilde d=d_0\vee \VaR_{\Lambda(x)}(X)\wedge \VaR_{\Lambda'(y)}(X)$, that is  better than any other $f\in\mathcal F$.  Thus, we have reduced the infinite-dimensional problem of finding a function $f$ to minimize  \eqref{eq:prob1} to the one-dimensional problem of finding the optimal value of  $d$. That is, 
$$\begin{aligned}
&\inf_{f\in\F,\VaR_{\Lambda'(y)} (f(X)) >y } \left \{ \VaR_{\Lambda(x)}(X)-\VaR_{\Lambda(x)}(f(X)) +(1-\theta)\E[f(X)]+\theta \VaR_{\Lambda'(y)} (f(X))    \right \}\\ &=\inf_{\substack{\VaR_{\Lambda(x)}(X)\leq d\le \VaR_{\Lambda'(y)} (X)\\d<\VaR_{\Lambda'(y)} (X)+  \VaR_{\Lambda(x)} (X)-y}} \left\{ (1-\theta)\E[X\wedge \VaR_{\Lambda(x)}(X)+(X- d)_+\wedge (\VaR_{\Lambda'(y)}(X) -d)]\right.\\&\quad \left.+\theta( \VaR_{\Lambda'(y)} (X)-d+  \VaR_{\Lambda(x)} (X) )\right\}.\end{aligned} $$
Denote by $$L_1(d)=(1-\theta)\E[(X- d)_+\wedge (\VaR_{\Lambda'(y)}(X) -d)] +\theta( \VaR_{\Lambda'(y)} (X)-d+  \VaR_{\Lambda(x)} (X) ) .$$
 It is clear that $L_1(d)$ is continuous in $d$ and its first-order derivative is given by $L'_1(d)=(1-\theta) F(d)-1<0.$ 
Thus, $d^*=\VaR_{\Lambda'(y)} (X)+  (\VaR_{\Lambda(x)} (X)-y)_-  $ minimizes $L_1(d)$. Therefore, we have 
 $$f^*(w)=w\wedge \VaR_{\Lambda(x)}(X) +(w-d^* )_+\wedge (y- \VaR_{\Lambda(x)}(X))_+,$$
and thus
$$\begin{aligned}&G_1(x,y)\\
&=\inf_{f\in\F,\VaR_{\Lambda'(y)} (f(X)) >y }  \left \{  \VaR_{\Lambda(x)}(X)-\VaR_{\Lambda(x)}(f(X)) +(1-\theta)\E[f(X)]+\theta \VaR_{\Lambda'(y)} (f(X))\right\}\\&= (1-\theta)\E[X\wedge \VaR_{\Lambda(x)}(X) +(X-d^* )_+\wedge (y- \VaR_{\Lambda(x)}(X))_+ ]  \\
&\quad +\theta (\VaR_{\Lambda(x)}(X)+(y-\VaR_{\Lambda(x)}(X))_+).   \end{aligned} $$
\\Case 2: \underline{$\Lambda(x)>\Lambda'(y)$}

\begin{itemize} \item If $S_X(w)> 1-\Lambda'(y)$, that is, $w<\VaR_{\Lambda'(y)}(X)\leq \VaR_{\Lambda(x)}(X)$,  we  have $H_1(w)=(1-\theta)S_X(w)-1+(\theta-\lambda).$ It is obvious that  $H_1(w)< 0$. 
\item If $S_X(w)\leq 1-\Lambda(x)$, that is, $w\ge \VaR_{\Lambda(x)}(X)$,  we have $H_1(w)=(1-\theta)S_X(w)\ge 0$.
\item If $1-\Lambda(x)<S_X(w)\leq 1-\Lambda'(y)$, that is, $\VaR_{\Lambda'(y)}(X) \leq w< \VaR_{\Lambda(x)}(X)$,  we have $H_1(w)=(1-\theta)S_X(w)-1 < 0$.
\end{itemize}	
Thus,  $H_1(w)<0$ for 	$w\in[0,\VaR_{\Lambda(x)}(X))$ 
and  $H_1(w)>0$ for 	$w\in[\VaR_{\Lambda(x)}(X),\infty)$.   In this case, we always have  $f^*(w)=w\wedge \VaR_{\Lambda(x)}(X)$, and $$\VaR_{\Lambda'(y)} (f^*(X))=\VaR_{\Lambda'(y)} (X\wedge \VaR_{\Lambda(x)} (X))= \VaR_{\Lambda'(y)} (X).$$
It then follows that  $$\begin{aligned}
&G_2(x,y)\\
=&\inf_{f\in\F,\VaR_{\Lambda'(y)} (f(X)) >y }   \left \{ \VaR_{\Lambda(x)}(X)-\VaR_{\Lambda(x)}(f(X)) +(1-\theta)\E[f(X)]+\theta \VaR_{\Lambda'(y)} (f(X))\right\}\\
= & \left\{\begin{array}{ll}
    (1-\theta)\E[X\wedge \VaR_{\Lambda(x)}(X)] +\theta \VaR_{\Lambda'(y)}(X), & \text{if } \VaR_{\Lambda'(y)}(X)\ge y, \\
    \infty, & \text{if } \VaR_{\Lambda'(y)}(X)< y.
\end{array}\right.\end{aligned}$$


{\bf Step 2:} Note that  \begin{equation}\label{eq:prob2}\begin{aligned} &\inf_{f\in\F,\VaR_{\Lambda'(y)} (f(X)) \leq y }   \left \{ \VaR_{\Lambda(x)}(X)-\VaR_{\Lambda(x)}(f(X)) +(1-\theta)\E[f(X)]+\theta  y   \right \}\\&= \inf_{f\in\F }   \left \{ \VaR_{\Lambda(x)}(X)-\VaR_{\Lambda(x)}(f(X)) +(1-\theta)\E[f(X)]+\theta  y +\lambda(\VaR_{\Lambda'(y)} (f(X)) - y)  \right \}\\
&= \int_0^\infty h(S_X(w))\d w + \int_0^\infty ((1-\theta)S_X(w) - h_1(S_X(w))+\lambda h_2(S_X(w)))q(w)\d w +\theta y-\lambda y)\\
&=: g_2(x,y). \end{aligned}\end{equation}
It is clear that  the following $q_2$ will minimize $g_2$: 
\begin{equation*}q_2(w)=\left\{\begin{aligned}&0,~~~~~&\text{if }(1-\theta)S_X(w) - h_1(S_X(w))+\lambda h_2(S_X(w))>0,\\&1,~~~~~&\text{if }(1-\theta)S_X(w) - h_1(S_X(w))+\lambda h_2(S_X(w))<0,\\& c(w),~~~&\text{otherwise},\end{aligned}\right.\end{equation*}
where $c$ could be any $[0,1]$-valued {Lebesgue-measurable} function on $\{w:(1-\theta)S_X(w) - h_1(S_X(w))+\lambda h_2(S_X(w))=0\}$.  Define $$H_2(w)=(1-\theta)S_X(w) - h_1(S_X(w))+\lambda h_2(S_X(w)).$$\\Case 1: \underline{$\Lambda(x)\leq\Lambda'(y)$}
\begin{itemize} \item If $S_X(w)> 1-\Lambda(x)$, that is, $w<\VaR_{\Lambda(x)}(X)\leq \VaR_{\Lambda'(y)}(X)$,  we  have $H_2(w)=(1-\theta)S_X(w)-1+\lambda.$ Let $d_0=\VaR_{\theta_0}(X)$ with $\theta_0=(\lambda-\theta)/(1-\theta)$. Then if $w>d_0$, we have $H_2(w)<0$ and if $w<d_0$, we have $H_2(w)>0$.

\item If $S_X(w)\leq 1-\Lambda'(y)$, that is, $y\ge \VaR_{\Lambda'(y)}(X)$,  we have $H_2(w)=(1-\theta)S_X(w)\ge 0$.
\item If $1-\Lambda'(y)<S_X(w)\leq 1-\Lambda(x)$, that is, $\VaR_{\Lambda(x)}(X) \leq w< \VaR_{\Lambda'(y)}(X)$,  we have $H_2(w)=(1-\theta)S_X(w)+\lambda \ge 0$.
\end{itemize}	
 By the above analysis, we can always construct a strategy  in the form of $f(w)=(w-d )_+\wedge(\VaR_{\Lambda(x)}(X) -d) $ with $ d\leq \VaR_{\Lambda(x)}(X)$ that is better than any other $f\in\mathcal F$. 
 \vspace{3mm} 
\\
Case 2: \underline{$\Lambda(x)>\Lambda'(y)$}
\begin{itemize} \item If $S_X(w)> 1-\Lambda'(y)$, that is, $w<\VaR_{\Lambda'(y)}(X)\leq \VaR_{\Lambda(x)}(X)$,  we  have $H_2(w)=(1-\theta)S_X(w)-1+\lambda.$   Let $d_0=\VaR_{\theta_0}(X)$ with $\theta_0=(\lambda-\theta)/(1-\theta)$.Then if $w>d_0$, we have $H_2(w)<0$ and if $w<d_0$, we have $H_2(w)>0$.

  \item If $S_X(w)\leq 1-\Lambda(x)$, that is, $w\ge \VaR_{\Lambda(x)}(X)$,  we have $H_2(w)=(1-\theta)S_X(w)\ge 0$. 

\item If $1-\Lambda(x)<S_X(w)\leq 1-\Lambda'(y)$, that is, $\VaR_{\Lambda'(y)}(X) \leq w< \VaR_{\Lambda(x)}(X)$,  we have $H_2(w)=(1-\theta)S_X(w)-1 < 0$.
\end{itemize}	
Thus,  $f(w)=(w-d)_+\wedge(\VaR_{\Lambda(x)}(X) -d)$ with $d\leq \VaR_{\Lambda'(y)}(X) $.

Combining the above two cases, we have 
$$\begin{aligned}
&\inf_{f\in\F,\VaR_{\Lambda'(y)} (f(X)) \le y }   \left \{ \VaR_{\Lambda(x)}(X)-\VaR_{\Lambda(x)}(f(X)) +(1-\theta)\E[f(X)]+\theta y   \right \} \\ =&\inf_{\substack{  d<\VaR_{\Lambda'(y)}(X) \wedge\VaR_{\Lambda(x)}(X) \\d\geq \VaR_{\Lambda'(y)}(X) \wedge\VaR_{\Lambda(x)}(X) -y}} \left\{ d+(1-\theta) \E[(X-d  )_+\wedge(\VaR_{\Lambda(x)}(X) -d]+\theta y\right\}.\end{aligned}$$
 Denote by $$L_2(d)=(1-\theta)\E[(X- d)_+\wedge (\VaR_{\Lambda(x)}(X) -d)] + d.$$ By the first order condition, we have $L'_2(d)=F(d)>0$, which implies that  $d^*=\max(\VaR_{\Lambda(x)}(X)\wedge  \VaR_{\Lambda'(y)}(X) -y,0)$, and thus $$f^*(w)= (w-(\VaR_{\Lambda(x)}(X)\wedge  \VaR_{\Lambda'(y)}(X)-y)_+ )_+\wedge(\VaR_{\Lambda(x)}(X)-( \VaR_{\Lambda(x)}(X)\wedge  \VaR_{\Lambda'(y)}(X)-y)_+ ).$$  Note that if $y> \VaR_{\Lambda(x)}(X)\wedge  \VaR_{\Lambda'(y)}(X)$,  we have $f^*(w)=w\wedge  \VaR_{\Lambda(x)}(X)$, then $$
 \begin{aligned}G_3(x,y)&=\inf_{f\in\F,\VaR_{\Lambda'(y)} (f(X)) <y }   \left \{ \VaR_{\Lambda(x)}(X)-\VaR_{\Lambda(x)}(f(X)) +(1-\theta)\E[f(X)]+\theta y   \right \}\\ &= (1-\theta) \E[(X\wedge\VaR_{\Lambda(x)}(X)]+y; \end{aligned}$$  otherwise,  for   $y\leq \VaR_{\Lambda(x)}(X)\wedge  \VaR_{\Lambda'(y)}(X)$,  we have $$
 \begin{aligned}G_3(x,y)&=\inf_{f\in\F,\VaR_{\Lambda'(y)} (f(X)) <y }   \left \{ \VaR_{\Lambda(x)}(X)-\VaR_{\Lambda(x)}(f(X)) +(1-\theta)\E[f(X)]+\theta y   \right \}\\ &= (1-\theta)(\VaR_{\Lambda(x)}(X)\wedge  \VaR_{\Lambda'(y)}(X) -y) + \theta  (\VaR_{\Lambda(x)}(X)\wedge  \VaR_{\Lambda'(y)}(X) ) \\&~~~+(1-\theta) \E[(X-(\VaR_{\Lambda(x)}(X)\wedge  \VaR_{\Lambda'(y)}(X)-y) )_+\wedge\\&~~~~(\VaR_{\Lambda(x)}(X)- (\VaR_{\Lambda(x)}(X)\wedge  \VaR_{\Lambda'(y)}(X)-y )] .   \end{aligned}$$

Now that the inner optimization problem has been solved,  we can proceed to solve the outer optimization problem concerning $y$, based on the results from Steps 1 and 2. {Specifically, we have the following expression: $$\begin{aligned}&\inf_{f\in\F}\Lambda\VaR(T_f(X)) \\&= \inf_{x\in\R_+}  \left\{\left\{ \inf_{ y\in \R_+,  \Lambda(x)\leq\Lambda'(y)}  \left \{  G_1(x,y) \right\}   \wedge  \inf_{  y  \leq\VaR_{\Lambda'(y)}(X),   \Lambda(x)>\Lambda'(y)} \left\{G_2(x,y)  \right \}  \wedge \inf_{y\in\R_+}  \left \{  G_3(x,y)  \right \}  \right\}\vee x \right\}.\end{aligned}$$   
 It is evident that $G_1(x,y)$    increases in $y$, while $G_2(x,y)$   decreases in $y$. Thus,  if there exists $y$ such that $\Lambda(x)\leq \Lambda'(y)$, we have $$ \inf_{ y\in\R_+, \Lambda(x)\leq\Lambda'(y)}    G_1(x,y)   = G_1(x,0)=(1-\theta)\E(X\wedge \VaR_{\Lambda(x)}(X) ) +\theta \VaR_{\Lambda(x)}(X), $$ otherwise,  $\inf_{ y\in\R_+,\Lambda(x)\leq\Lambda'(y)}    G_1(x,y) =\infty.$ Similarly,  if there exists $y$ such that   $\Lambda(x)>\Lambda'(y)$, then $$\inf_{y \leq \VaR_{\Lambda'(y)(X)}, \Lambda(x)>\Lambda(y)}   G_2(x,y)   =G_2(x, \Lambda'\VaR(X))=(1-\theta)\E(X\wedge \VaR_{\Lambda(x)}(X) ) +\theta \Lambda'\VaR(X), $$ otherwise,  $\inf_{y \leq \VaR_{\Lambda'(y)(X)}, \Lambda(x)>\Lambda(y)}   G_2(x,y)=\infty$.}
Next, consider the minimization problem for $G_3$. If $y> \VaR_{\Lambda(x)}(X)\wedge  \VaR_{\Lambda'(y)}(X)$,  then $$\begin{aligned}\inf_{y> \VaR_{\Lambda(x)}(X)\wedge  \VaR_{\Lambda'(y)}(X)}   G_3(x,y)   &=G_3(x, \VaR_{\Lambda(x)}(X)\wedge  \Lambda'\VaR(X))\\&=(1-\theta) \E[(X\wedge\VaR_{\Lambda(x)}(X)]+\VaR_{\Lambda(x)}(X)\wedge    \Lambda'\VaR(X).\end{aligned}$$
  Otherwise, with a slight abuse of notation, let $l=\VaR_{\Lambda(x)}(X)\wedge  \VaR_{\Lambda'(y)}(X) -y$ and   $$\bar G_3(l)=(1-\theta) \left\{l+\E[(X-l)_+\wedge( \VaR_{\Lambda(x)}(X)-l)\right\} .$$  It is clear that $\bar  G'_3(l)=F(l)>0$, which implies that $G_3(x,y)$ decreases in $y$. In particular, when $y=\VaR_{\Lambda(x)}(X)\wedge  \VaR_{\Lambda'(y)}(X)$,  that is $y=\VaR_{\Lambda(x)}(X)\wedge  \Lambda'\VaR(X)$, we have $$\begin{aligned}\inf_{y\leq  \VaR_{\Lambda(x)}(X)\wedge  \VaR_{\Lambda'(y)}(X)}  \left \{  G_3(x,y)  \right \} &=G_3(x, \VaR_{\Lambda(x)}(X)\wedge  \Lambda'\VaR(X))\\& =  (1-\theta) \E[X\wedge\VaR_{\Lambda(x)}(X)]+ \theta  (\VaR_{\Lambda(x)}(X)\wedge  \Lambda'\VaR(X)).\end{aligned} $$
To summarize, we have the following result
\begin{equation}\label{eq:outer}\begin{aligned}\inf_{f\in\F}\Lambda\VaR(T_f(X)) =\inf_{x\in\R_+}  \left\{ \left \{ (1-\theta)\E[  X\wedge \VaR_{\Lambda(x)}(X)]+\theta (\Lambda'\VaR ( X)\wedge \VaR_{\Lambda(x)}(X) )   \right \} \vee x \right\}.
\end{aligned}
\end{equation}
Let $$x^*=\inf\{x\in\R_+: (1-\theta)\E[  X\wedge \VaR_{\Lambda(x)}(X)]+\theta (\Lambda'\VaR ( X) \wedge \VaR_{\Lambda(x)}(X))\leq x\}.$$
{ Similarly to the proof of Theorem \ref{thm:1}, the function $$x\mapsto (1-\theta)\E[  X\wedge \VaR_{\Lambda(x)}(X)]+\theta (\Lambda'\VaR ( X) \wedge \VaR_{\Lambda(x)}(X))$$ is decreasing in $x$. As $\Lambda(0)<1$ and $\p(X<\infty)=1$, we have $x^*<\infty$. Because $\Lambda$ is right-continuous, the infimum in \eqref{eq:outer} is attainable and thus
$$x^*=\min_{x\in\R_+}  \left\{ \left \{ (1-\theta)\E[  X\wedge \VaR_{\Lambda(x)}(X)]+\theta (\Lambda'\VaR ( X)\wedge \VaR_{\Lambda(x)}(X) )   \right \} \vee x \right\},$$
which is obtained at $x=x^*$.}
Therefore, we have
$\inf_{f\in\F}\Lambda\VaR(T_f(X))=x^*.$
This completes the proof.
\end{proof}  
Equation \eqref{eq:f4} shows that the ceded loss function takes the form of  a limited loss transform. Under this strategy,  the DM retains any loss exceeding the  threshold $\VaR_{\Lambda(x^*)}(X)$. While this threshold is determined by the function $\Lambda$,  the value of $x^*$ in \eqref{eq:x4} is  determined based on $\Lambda'\VaR$.  Notably, as $\Lambda' \VaR(X)$ increases, $x^*$ also increases, which reduces the threshold. Consequently, a higher premium results in the DM retaining a larger portion of the claim. Further numerical analysis can be found in Section \ref{sec:examples}. 
 
The following corollary directly follows from Theorem \ref{thm:3}. 
\begin{corollary}  
If $\Lambda(x) \equiv \alpha$ and $\Lambda'(y)\equiv\beta$ for any $x,y\ge 0$, where  $\alpha,\beta \in (0,1)$, the solution to \eqref{eq:obj_3}  is given by 
$$f^*(w)=w\wedge\VaR_\alpha(X),$$ and thus $$\inf_{f\in\F}\VaR_\alpha(T_f(X))= (1-\theta) \E[X\wedge\VaR_\alpha(X)]+\theta(\VaR_\beta(X)\wedge\VaR_\alpha(X)).
$$\end{corollary}
\section{Optimal insurance  under model uncertainty}\label{sec:model-uncertainty}
In the above sections, we assumed that the distribution of the underlying loss random variable is either given or known. However, this assumption is sometimes unrealistic because the actual loss distribution is generally unavailable in most practical scenarios. This recognition prompts us to explore model uncertainty in the determination of optimal insurance policies. In the following context, we investigate an optimal insurance problem incorporating a likelihood ratio uncertainty  or assuming that the loss is partially known in the sense that only the first two moments of the loss are available.  The  closed-form solutions are obtained for both uncertainty scenarios.  The premium is calculated using \eqref{eq:ev_p}, which is based on the expected value premium principle.    

 \subsection{Likelihood ratio uncertainty}
Suppose the DM's reference probability measure $\mathbb{P}$ is fixed, and $\Lambda\VaR$ is evaluated under $\mathbb{P}$. 
However, the true probability measure  and the corresponding distribution of \( X \) might not be known. To account for this uncertainty,  we    use the likelihood ratio to describe the DM's attitude towards uncertainty in probability measures, as introduced in \cite{LMWW22} and \cite{LLS24}. 

 We define \( \mathcal{P} \) as the collection of probability measures that are absolutely continuous with respect to \( \mathbb{P} \). The uncertainty set is defined as
$$
\mathcal{P}_\beta=\left\{\mathbb{Q} \in \mathcal{P} : \frac{\mathrm{d} \mathbb{Q}}{\mathrm{d} \mathbb{P}} \leq \frac{1}{\beta}\right\},
$$
for a given $\beta \in (0,1]$ representing the degree of uncertainty faced by the DM. 
 A larger value of $\beta$ corresponds to a  lower degree of uncertainty, with $\beta = 1$ indicating that  the DM encounters no uncertainty.
 
 The DM aims to solve the following robust optimization problem  
\begin{equation}\label{eq:robust1}
\min _{f \in \mathcal{F}} \sup _{\mathbb{Q} \in \mathcal{P}_\beta} \Lambda\VaR^{\mathbb{Q}}(T_f),
\end{equation}
where $\Lambda\VaR^{\mathbb{Q}}$ is $\Lambda\VaR$ evaluated under the probability measure $\mathbb{Q}$ instead of $\mathbb{P}$. 

\begin{lemma}[Proposition 4 of \citealp{XH24}]  \label{lem:4}Let $\Lambda: \mathbb{R} \rightarrow[0,1]$ be decreasing. For $\beta \in(0,1]$, define $\Lambda_\beta=\beta\Lambda+1-\beta $. Then
$$
\sup _{\mathbb{Q} \in \mathcal{P}_\beta} \Lambda\VaR^{\mathbb{Q}}(X) = \Lambda_\beta\VaR^{\mathbb{Q}}(X) , \quad X \in \X.
$$
\end{lemma} 

The following theorem is readily derived by combining Lemma \ref{lem:4} and Theorem \ref{thm:1}. Specifically, it can be obtained by substituting \(\Lambda\) with \(\Lambda_\beta\) throughout Theorem \ref{thm:1}. Recall that  $d^*=\VaR_{\theta^*}(X)$, where $\theta^*=\theta/(1+\theta)$. 
	\begin{theorem}\label{thm:4} { Let $\Lambda\in\mathcal D$ and suppose that $\Lambda$ is right-continuous.} For $\beta \in(0,1]$, define $\Lambda_\beta=\beta\Lambda+1-\beta $.  The optimal $f^{*}$ that solves the $\Lambda\VaR$-based insurance model  \eqref{eq:robust1} over the class of ceded loss functions $\F$ is given by
$$
f^{*}(x)= \begin{cases}\min \left\{\left(x-d^*\right)_{+}, \VaR_{\Lambda_\beta(x^*)}(X)-d^*\right\}, &\text{if}~d^*\le \VaR_{\Lambda_\beta(x^*)}(X),  \\ 0, & \text { otherwise },\end{cases}
$$ in which $x^*=\inf\{x\in\R_+: G(x)\leq x\}$ with 
\begin{equation*}G(x)= d^* \wedge \VaR_{\Lambda_\beta(x)}(X)   +(1+\theta) \mathbb{E}\left[\min \left\{\left(X-d^*\right)_{+},\left(\VaR_{\Lambda_\beta(x)}(X)-d^*\right)_{+}\right\}\right].\end{equation*} 
Moreover, 
$\Lambda\VaR^{\mathbb{Q}^*}(T_{f^*}(X))=x^*.$
\end{theorem}
Note that $\Lambda_\beta = \Lambda + (1 - \Lambda)(1 - \beta) \geq \Lambda$. In particular, if $\beta = 1$, then the robust optimization problem \eqref{eq:robust1} reduces to \eqref{eq:obj_general}. As $\beta$ decreases, $\Lambda_\beta$ clearly increases, indicating that the DM applies a higher overall confidence level function when determining the optimal insurer strategy. Furthermore, from the equation for solving $x^*$, it shows that $x^*$ increases as $\beta$ decreases, due to the increases in $\Lambda_\beta$. Since $\Lambda_\beta$ is a decreasing function of $x$, this results in a trade-off between risk and uncertainty for the DM when using $\Lambda\VaR$ as the measure. We provide a numerical example of \(\Lambda(x^*)\) related to \(\beta\) in Section \ref{sec:examples}. 
\begin{remark} If we assume that the contract follows a stop-loss form as in Section \ref{sec:3.2}, a solution similar to that in Theorem \ref{thm:deductbile} can also be obtained. We only need to replace all instances of $\Lambda$ with $\Lambda_\beta$ in Theorem \ref{thm:deductbile}, so we do not repeat the details here.\end{remark}
We immediately have the following corollary, which  is  also shown in Example  3.1 of  \cite{LLS24}. 
 
\begin{corollary}\label{cor:3}If $\Lambda(x)\equiv\alpha$ for any $x\in\R_+$, where $\alpha\in(0,1)$,    the optimization problem  \eqref{eq:robust1}   can be solved by $$f^*(x)= \begin{cases} \min\{(x-d^*)_{+}, \VaR_{\beta\alpha+1-\beta}(X)-d^*\}, & \text{if}~d^*\leq \VaR_{\beta\alpha+1-\beta}(X)\\0, & \text{otherwise},\end{cases}$$  and $$
\begin{aligned}&\min _{f \in \mathcal{F}} \sup _{\mathbb{Q} \in \mathcal{P}_\beta} \VaR_\alpha^{\mathbb{Q}}(T_f) \\&=d^* \wedge \VaR_{\beta\alpha+1-\beta}(X)   +(1+\theta) \mathbb{E}\left[\min \left\{\left(X-d^*\right)_{+},\left(\VaR_{\beta\alpha+1-\beta}(X)-d^*\right)_{+}\right\}\right].\end{aligned}$$
\end{corollary}
Note that $\beta\alpha + 1 - \beta = \alpha + (1 - \alpha)(1 - \beta) \geq \alpha$.  Corollary \ref{cor:3} shows that the deductible in the worst-case scenario is greater than that in  Corollary \ref{cor:1} where model uncertainty is not involved, which is a reasonable outcome given the increased risk exposure. This behavior reflects a risk-averse strategy: when faced with  uncertainty about future outcomes, the DM opts for more insurance coverage to safeguard against potential adverse events.
Additionally, as $\beta$ decreases, $\beta\alpha + 1 - \beta$ increases, leading to a larger $\VaR_{\beta\alpha + 1 - \beta}(X)$. This finding suggests that as ambiguity in the probability measure increases (indicated by a higher ${1}/{\beta}$), the DM is more likely to choose higher insurance coverage. 
\subsection{Mean-variance uncertainty}
In this subsection, we assume that the loss is partially known in the sense that only the first two moments of $X$ are known and finite. Let $\M_+$ be the set of CDFs of random variables in $\X_+$. Given a pair of non-negative mean and standard deviation $(\mu, \sigma)$ of $X$, we define the uncertainty set 
\begin{equation}\label{eq:mv_set} \mathcal{S}(\mu, \sigma)=\left\{F\in \M_+:   \int_0^{\infty} x \mathrm{~d} F(x)=\mu, \int_0^{\infty} x^2 \mathrm{~d} F(x)=\mu^2+\sigma^2\right\}.\end{equation}
The $\mathcal{S}(\mu, \sigma)$ specifies a set of CDFs with the same mean $\mu$ and variance $\sigma^2$.   
The distributionally robust insurance problem is defined by 
\begin{equation}\label{eq:R_vf}
    \inf_{f\in\mathcal F} \sup_{F\in \mathcal{S}(\mu, \sigma)} \Lambda\VaR(T_f(X)),
\end{equation} 
where  $T_f$ is given by \eqref{eq:T_f}. 

For a set of distributions $\M$, we denote $G_{\M}^-(x)=\inf_{G\in\M}G(x)$, for $x\in\R_+$.  In the following context, we abuse the notation by treating $\Lambda\VaR$ also as a mapping from a set $\mathcal M$ of distributions to $\R$; that is, we write $\Lambda\VaR(G) =\Lambda\VaR (Y)$ when $Y$ follows the distribution $G$.

The following lemma plays a key role in solving the problem \eqref{eq:R_vf}.
\begin{lemma}[Theorem 2 of \citealp{HL24}]  \label{lem:3}Let $\M$ be a set of distributions and $\Lambda\in\mathcal D$. 
We have  $\sup_{G\in\mathcal M}\Lambda\VaR(G)=\Lambda\VaR (G_{\mathcal M}^-)$.
\end{lemma}

\begin{proposition}\label{prop:4} For a set of distributions \(\mathcal{M}\) that includes the candidate distribution of \(X\) { and $\Lambda\in\mathcal D$}, we have
$$
\inf_{f \in \mathcal{F}} \sup_{F \in \mathcal{M}} \Lambda\VaR(T_f(X)) = \inf_{x \in \mathbb{R}_+} \left\{ \inf_{f \in \mathcal{F}} \sup_{F \in \mathcal{M}} \VaR_{\Lambda(x)}(T_f(X)) \vee x \right\}.
$$
\end{proposition}
\begin{proof} By Lemma \ref{lem:3}, for $d\geq0$, we have 
$$   \sup_{F\in \mathcal M} \Lambda\VaR(T_f(X)) = \sup_{G\in \mathcal M _f} \Lambda\VaR(G)=\Lambda\VaR (G_{\M_f}^-),$$
where $$ \mathcal{M}_f=\{G:[0,\infty)\to [0,1]:  G ~\text{is the distribution function of} ~T_f,  F\in \mathcal M\}.$$
Because  $\sup_{F\in \mathcal M}  \VaR_{\Lambda(x)}(T_f(X))=\sup_{G\in\M_f}\VaR_{\alpha}(G)=\VaR_{\alpha}(G_{\M_f}^-)$ \citep[see Corollary 1 of][]{HL24}, by Lemma \ref{lem:representation},  we have 
$$\begin{aligned}
\inf_{f\in\mathcal F} \sup_{F\in \mathcal M}  \Lambda\VaR(T_f(X))=\inf_{f\in\mathcal F}   \Lambda\VaR (G_{\M_f}^-)&=  \inf_{f\in\mathcal F} \inf_{x\in\R_+} \left\{\VaR_{\Lambda(x)} (G_{\M_f}^- ) \vee x\right\}\\
&=  \inf_{x\in\R_+} \inf_{f\in\mathcal F}\left\{\sup_{F\in \mathcal M}  \VaR_{\Lambda(x)}(T_f(X))  \vee x\right\}\\
&=  \inf_{x\in\R_+} \left\{\inf_{f\in\mathcal F}\sup_{F\in \mathcal M}  \VaR_{\Lambda(x)}(T_f(X))  \vee x\right\},  \end{aligned}
 $$which completes the proof.
\end{proof} In fact, Proposition \ref{prop:4} shows that  we can switch the order of the supremum and infimum in the problem \eqref{eq:R_vf} as follows: 
$$\inf_{f\in\mathcal F}\sup_{F\in \mathcal M}\inf_{x\in\R_+}   \left\{\VaR_{\Lambda(x)}(T_f(X))  \vee x\right\}= \inf_{x\in\R_+}   \left\{ \inf_{f\in\mathcal F}\sup_{F\in \mathcal M} \VaR_{\Lambda(x)}(T_f(X))  \vee x\right\}.$$
 It is important to note that the conclusion in Proposition \ref{prop:4} is independent of specific forms of \(\mathcal{M}\); it does not need to be the uncertainty set defined by \(\mathcal{S}_{\mu, \sigma}\) in \eqref{eq:mv_set}. It can also  be applied to other types of uncertainty sets, such as the Wasserstein uncertainty set or risk aggregation sets; see Section 4 of \cite{HL24}. Also, for \(\alpha \in (0,1)\), the optimization problem 
\begin{equation}\label{eq:worst}\inf_{f \in \mathcal{F}} \sup_{F \in \mathcal{M}} \VaR_{\alpha}(T_f(X))\end{equation} represents the robust insurance problem with VaR. This problem has been studied in the literature; see e.g., \cite{LM22} and \cite{CLY24}.

As in Section \ref{sec:3.2}, we restrict the set of admissible insurance indemnities to be the set of  stop-loss contracts, that is,  we assume $f(x)=(x-l)_+$ for some $l\geq0$.  Thus, the insurer's optimization problem \eqref{eq:R_vf}  becomes 
\begin{equation}\label{eq:robust2}
\min _{l\geq0} \sup _{F\in \mathcal{S}_{\mu,\sigma}} \Lambda\VaR(T_l(X)),
\end{equation}
where $T_l$ is defined in \eqref{eq:T_l}.

For the special case where \(\Lambda(x) \equiv \alpha\) for any \(x \in \mathbb{R}_+\), the optimization problem \eqref{eq:robust2} reduces to the robust insurance problem with VaR, which has been addressed by Theorem 2 of \cite{LM22}.

\begin{lemma}[Theorem 2 of \citealp{LM22}] \label{lem:5} Let $\Lambda(x) \equiv \alpha\in(0,1)$. 
 If  $\theta \leq \sigma^2/\mu^2$, then an optimal deductible of the optimization problem \eqref{eq:robust2} is
$
l^*= \infty \id_{\{\alpha< \theta^*\}},
$
and the optimal value is
$$
\inf_{l\geq0}\sup_{F\in \mathcal S (\mu,\sigma )} \VaR_{\alpha}(T_l(X))  = \begin{cases} (1+\theta) \mu, &\text{if~} \alpha\geq \theta^*, \\\frac{\mu}{1-\alpha}, & \text{if~}\alpha<\theta^* .\end{cases}
$$
If $\theta> \sigma^2/\mu^2$, then the optimal deductible of the optimization problem \eqref{eq:robust2}   is $$
l^*=\left(\mu-\sigma \frac{1-\theta}{2 \sqrt{\theta}}\right)\id_{\{\alpha\geq \theta^*\}}+ \infty\id_{\{\alpha< \theta^*\}},
$$ and the optimal value is
$$
\inf_{l\geq0}\sup_{F\in \mathcal S (\mu,\sigma )}\VaR_{\alpha}(T_l(X))  =\sup _{F \in \mathcal{S}(\mu, \sigma)} \VaR_\alpha(X)= \begin{cases}\mu+\sigma \sqrt{\theta}, & \text{if~}\alpha\geq\theta^*, \\ \mu+\sigma \sqrt{\frac{\alpha}{1-\alpha}}, &  \text{if~}\alpha<\theta^*.\end{cases}
$$
\end{lemma} 
Combining  Proposition \ref{prop:4} and  Lemma \ref{lem:5}, we have the following theorem. 
\begin{theorem}\label{thm:5}
{ Let $\Lambda\in\mathcal D$ and suppose that $\Lambda$ is right-continuous.}
If $\theta\leq \sigma^2/\mu^2$, the solution $l^*\in[0,\infty]$ that solves the robust  $\Lambda\VaR$-based insurance model   \eqref{eq:robust2}   is  
given by 
\begin{equation}\label{eq:d1}
l^*= \begin{cases} 0, &\text{if~} x_1^*=(1+\theta) \mu,  \\ \infty, &  
\text{if~} x_1^*<(1+\theta) \mu, \end{cases}
\end{equation}
and the corresponding minimum is given by $\Lambda\VaR\left(T_{l^*}(X)\right)=x_1^*$ with $x_1^*$  defined as
\begin{equation*}
x^*_1=\inf\left\{x\in\R_+: (1+\theta )\mu \id_{\{\Lambda(x)\geq \theta^* \}}  + \frac{\mu}{1-\Lambda(x)} \id_{\{\Lambda(x)< \theta^*  \}}\leq x\right\}. \end{equation*}

If $\theta> \sigma^2/\mu^2$, the solution $l^*\in[0,\infty]$ that solves the robust  $\Lambda\VaR$-based insurance model   \eqref{eq:robust2}   is  
given by 
\begin{equation}\label{eq:d2}
l^*= \begin{cases} \mu-\sigma \frac{1-\theta}{2 \sqrt{\theta}}, &\text{if~} x_2^*=\mu+\sigma \sqrt{\theta},  \\ \infty, &  
\text{if~} x_2^*<\mu+\sigma \sqrt{\theta}, \end{cases}
\end{equation}
and the corresponding minimum is given by $\Lambda\VaR\left(T_{l^*}(X)\right)=x_2^*$ with $x_2^*$  defined as
\begin{equation*}
x^*_2=\inf\left\{x\in\R_+: \left(\mu+\sigma \sqrt{\theta}\right) \id_{\{\Lambda(x)\leq \theta^* \}}  +\left( \mu+\sigma \sqrt{\frac{\Lambda(x)}{1-\Lambda(x)}}\right) \id_{\{\Lambda(x)> \theta^* \}} \leq x\right\}. 
\end{equation*}

 \end{theorem} 
 \begin{proof} Recall that $\theta^*=\theta/(1+\theta)$.   For  $x\in\R_+$,   $(1-\Lambda(x))(1+\theta) \leq 1$ is equivalent to  $\Lambda(x)\leq \theta^*$. By Proposition \ref{prop:4} and Lemma \ref{lem:5}, 
  if $\theta\leq \sigma^2/\mu^2$, 
 we have $$\begin{aligned}
 \inf_{l\geq0}  \sup_{F\in \mathcal S (\mu,\sigma )} \Lambda\VaR(T_l(X))&=\inf_{x\in\R_+}   \left\{ \inf_{l\geq0}\sup_{F\in \mathcal S (\mu,\sigma )} \VaR_{\Lambda(x)}(T_l(X))  \vee x\right\}\\
 &=\inf_{x\in\R_+} \left\{\left((1+\theta )\mu \id_{\{\Lambda(x)\geq \theta^* \}}  + \frac{\mu}{1-\Lambda(x)} \id_{\{\Lambda(x)<\theta^* \}}\right)\vee x \right\}.
 \end{aligned}$$
 Recall that
\begin{equation*}
x^*_1=\inf\left\{x\in\R_+: (1+\theta )\mu \id_{\{\Lambda(x)\geq \theta^* \}}  + \frac{\mu}{1-\Lambda(x)} \id_{\{\Lambda(x)< \theta^*  \}}\leq x\right\}. \end{equation*} 
{
It is clear that the function
$$x\mapsto (1+\theta )\mu \id_{\{\Lambda(x)\geq \theta^* \}}  + \frac{\mu}{1-\Lambda(x)} \id_{\{\Lambda(x)< \theta^*  \}}$$
is decreasing. Therefore, we have $x^*_1\le (1+\theta)\mu$ and
$$\inf_{l\geq0}  \sup_{F\in \mathcal S (\mu,\sigma )} \Lambda\VaR(T_l(X))=x^*_1.$$ Because $\Lambda$ is right-continuous, we have
$$x^*_1=\min\left\{x\in\R_+: (1+\theta )\mu \id_{\{\Lambda(x)\geq \theta^* \}}  + \frac{\mu}{1-\Lambda(x)} \id_{\{\Lambda(x)< \theta^*  \}}\leq x\right\}.$$
If $x^*_1=(1+\theta)\mu$, then $\Lambda(x^*_1)\ge\theta^*$. By Lemma \ref{lem:5}, we have $l^*=0$. If $x^*_1<(1+\theta)\mu$, then $\Lambda(x^*_1)<\theta^*$. Lemma \ref{lem:5} yields that $l^*=\infty$. Therefore, the optimal deductible $l^*$ for problem \eqref{eq:robust2} is given by \eqref{eq:d1}.
}
On the other hand,  if $\theta>\sigma^2/\mu^2$, we have   
$$\begin{aligned}  
&\inf_{l\geq0}  \sup_{F\in \mathcal S (\mu,\sigma )} \Lambda\VaR(T_l(X))\\
&=\inf_{x\in\R_+} \left\{\left(\left(\mu+\sigma \sqrt{\theta}\right) \id_{\{\Lambda(x)\geq \theta^* \}}  + \left(\mu+\sigma \sqrt{\frac{\Lambda(x)}{1-\Lambda(x)}}\right) \id_{\{\Lambda(x)< \theta^* \}}\right)\vee x \right\}.\end{aligned}$$
 Recall that
 \begin{equation*}
x^*_2=\inf\left\{x\in\R_+: \left(\mu+\sigma \sqrt{\theta}\right) \id_{\{\Lambda(x)\leq \theta^* \}}  +\left( \mu+\sigma \sqrt{\frac{\Lambda(x)}{1-\Lambda(x)}}\right) \id_{\{\Lambda(x)> \theta^* \}} \leq x\right\}. 
\end{equation*}
 { Similarly to the first case}, we have $x_2^*\leq\mu+\sigma \sqrt{\theta}$ and
 $$\inf_{l\geq0}  \sup_{F\in \mathcal S (\mu,\sigma )} \Lambda\VaR(T_l(X))=x^*_2.$$ 
 By { the right-continuity of $\Lambda$ and} Lemma \ref{lem:5}, the optimal deductible $l^*$ is given by \eqref{eq:d2}.
\end{proof}
By comparing the strategies in Lemma \ref{lem:5} and Theorem \ref{thm:5}, we observe that the forms of the optimal solutions are similar. However, for a VaR agent, the optimal strategy depends on the values of \(\alpha\) and \(\theta^*\). In contrast, for a \(\Lambda\VaR\) agent, determining the optimal strategy involves finding the optimal confidence level by solving for \(x^*_i\) \((i=1,2)\). Some observations that hold for a VaR agent can also be applied to a \(\Lambda\VaR\) agent. Specifically, when  \(\theta^*\) is higher than a certain confidence level (\(\alpha\) for the VaR agent and \(\Lambda(x_i^*)\) for the \(\Lambda\VaR\) agent), the DM always chooses to buy no insurance. Conversely, if $\theta^*$ is smaller than the confidence level, the DM either chooses to transfer all risk to the insurer or to purchase some insurance. Notably, in \eqref{eq:d2}, when \(x_2^* = \mu + \sigma \sqrt{\theta}\), the optimal deductible \(l^*\) increases with \(\theta\). This suggests that the insurer opts to retain more risk as the insurance premium rises. Additionally, if \(\theta > 1\), the deductible is larger than the mean of the claim due to the high cost of the premium.

\section{Numerical examples}\label{sec:examples}
In this section, we numerically investigate how the solutions to the optimal insurance problem with $\Lambda\VaR$  interact with changes in the underlying model parameters.  

{Assume that the random loss $X$ follows a Pareto distribution, which is commonly used in modeling large losses caused by catastrophic events. The Pareto distribution function with parameter $\alpha>0$ and $m>0$ is given by
$$\p(X\le x)=1-\left(\frac{m}{x}\right)^\alpha, ~~~x\ge m.$$
Here $\alpha>0$ is the tail parameter and $m$ is the scale parameter. The Pareto distribution with a smaller $\alpha$ has a heavier tail. For $\alpha>1$, $\E[X]=\alpha m/(\alpha-1)$, and for $\alpha\le 1$, $\E[X]=\infty$.  Let 
$$\Lambda(x)= \ell e^{-k x}+0.9, ~~x\in\R_+,$$
where $k\ge 0$ and $\ell \in[0,0.1)$.
Thus, a larger loss is tolerated with a probability level ranging from 0.9 to $0.9+\ell$,  decreasing as the loss increases. An agent with a larger $\ell$ is more concerned with the tail part of the risk and an agent with a larger $k$ is more tolerant of large losses.

We first consider problem \eqref{eq:obj_general} with the expected premium principle and general contracts. 
By Theorem \ref{thm:1}, the optimal ceded loss function is $
f^{*}(x)= \min \{(x-d^*)_{+}, \VaR_{\Lambda(x^*)}(X)-d^*\}
$, where $d^*$ depends on $\theta$ via $d^*=\VaR_{\theta/(1+\theta)}(X)$ and $x^*$, as in \eqref{eq:x_star}, can be calculated numerically. Fixing $\theta=0.25$, we compute $\Lambda(x^*)$ for different Pareto risks in cases where $k\in\{0.001,0.01,0.1\}$ and $\ell \in \{0.099,0.0999\}$. Assume that the expectations of these Pareto risks are equal to 1; that is, for a fixed $\alpha>1$, $m$ is taken to be $1-1/\alpha$. 

We plot $\Lambda(x^*)$ against $\alpha\in(1.25,2.6)$ in Figure \ref{fig:alpha}. When $k$ and $\ell$ are fixed, as $\alpha$ increases,  $\Lambda(x^*)$ decreases first and it may become increasing when $\alpha$ reaches a certain level. In other words, when $\alpha$ is small, for more heavy-tailed risks, more risk needs to be transferred under the optimal strategy. This may seem counterintuitive as the DM is more tolerant of large losses. To understand this intriguing phenomenon, it is important to note that as the Pareto risks have the same expectation,  one risk cannot stochastically dominate another, even if it has a heavier tail. Although heavy-tailed risks are more likely to take extremely large values, it is less probable for them to take moderately large values. Therefore, for DMs who care less about extremely large values of risks, less heavy-tailed Pareto risks can be more risky, leading to less risk being transferred. The observation may flip when $\alpha$ is large, i.e., less risk is transferred for more heavy-tailed risks. This is because in these cases the DM is more concerned about the extremely large values of the risks with larger $\alpha$ and thus more heavy-tailed risks are indeed more dangerous. When $\alpha$ is fixed, DMs with larger $\ell$ and smaller $k$  are more careful about extremely large values of risks, leading to more risk being transferred.

\begin{figure}[h!]
\centering
        \includegraphics[width=1\linewidth]{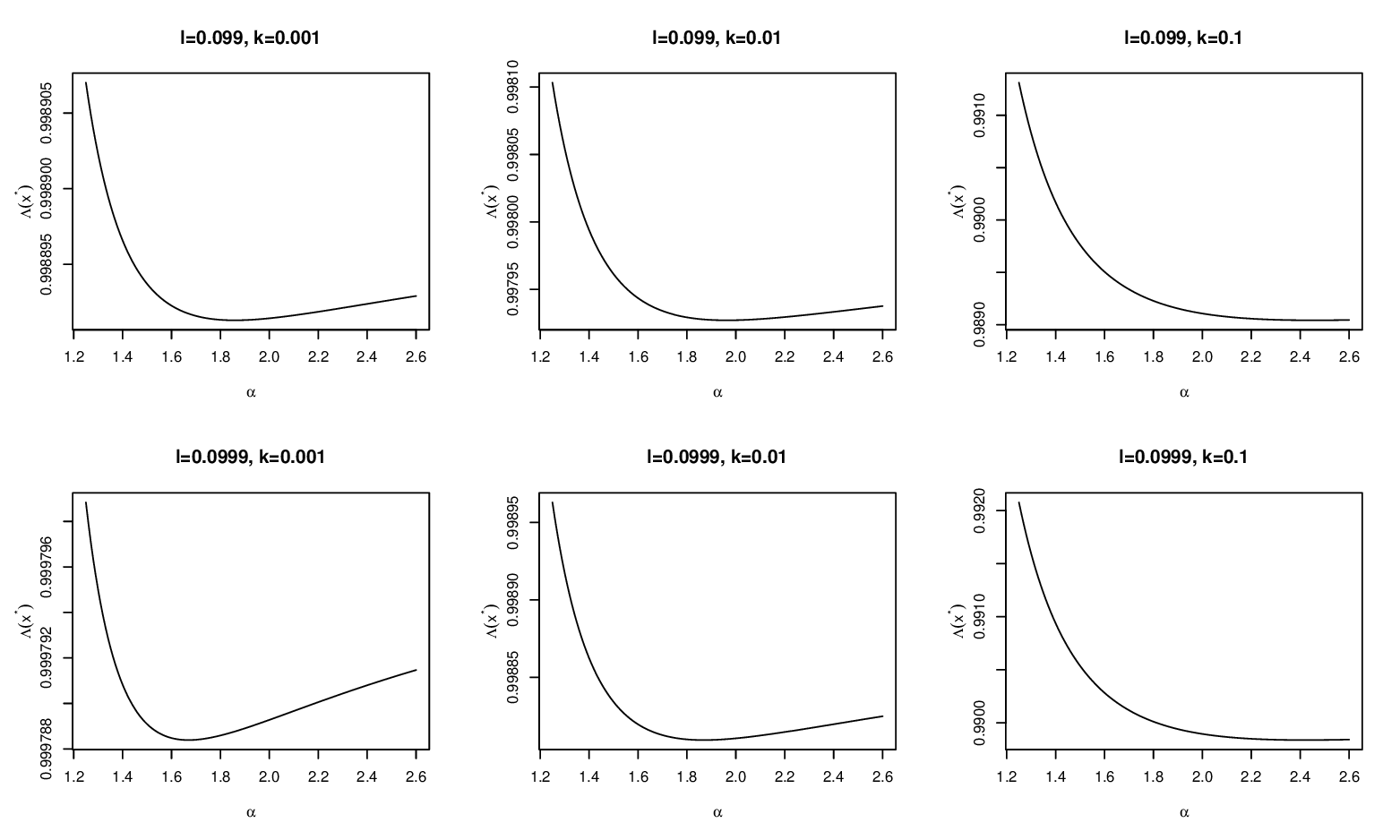}
        \caption{$\Lambda(x^*)$ for a range of values of $\alpha$  in problem \eqref{eq:obj_general} with the expected premium principle.}
        \label{fig:alpha}
\end{figure}

In Figure \ref{fig:theta},  fixing $\alpha=1.5$, $m=1/3$, $\ell=0.09$, and $k=0.1$, we plot $\Lambda(x^*)$ and $\Lambda(x^*)-\theta^*$ for a range of values  of $\theta$, where $\theta^*=\theta/(1+\theta)$. From the left plot of Figure \ref{fig:theta}, we find that less loss from the tail part will be transferred as $\theta$ gets larger. Moreover, less loss will be transferred at an overall level as well, confirmed by the right plot of Figure \ref{fig:theta}.

\begin{figure}[h!]
\centering
        \includegraphics[width=0.8\linewidth]{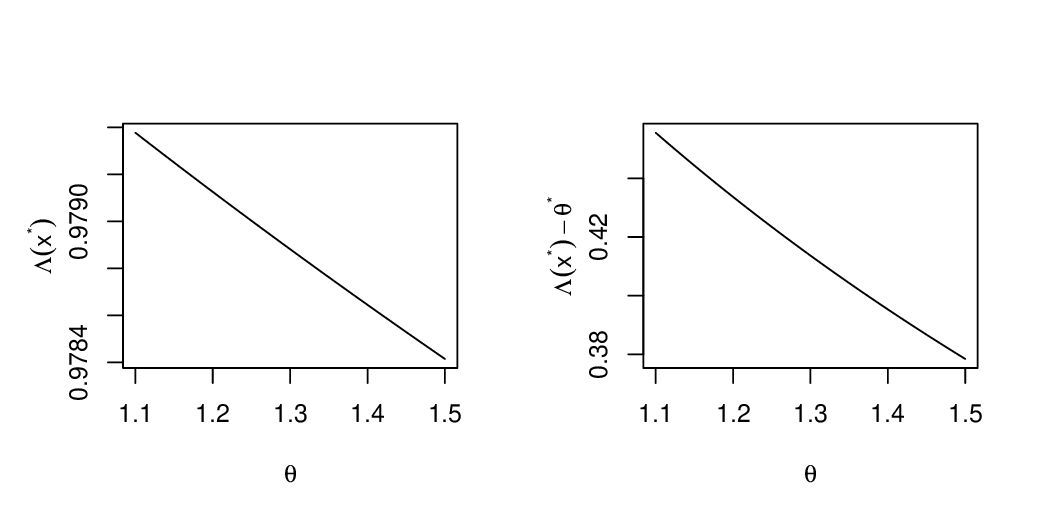}
       \caption{$\Lambda(x^*)$ and $\Lambda(x^*)-\theta^*$ for a range of values of $\theta$  in problem \eqref{eq:obj_general} with the expected premium principle.}
        \label{fig:theta}
\end{figure}
Next, we continue to study problem \eqref{eq:robust1} with the expected premium principle, under likelihood ratio uncertainty.  For $\beta \in(0,1]$, let $\Lambda_\beta=\beta\Lambda+1-\beta $. Here, $\beta$ indicates the level of uncertainty; the smaller $\beta$ is, the more model uncertainty we have.
By Theorem \ref{thm:4}, the optimal ceded loss function is $
f^{*}(x)= \min \{(x-d^*)_{+}, \VaR_{\Lambda_\beta(x^*)}(X)-d^*\}
$. Letting $\alpha=1.5$, $m=1/3$, and $\theta=0.25$, we display $\Lambda_\beta(x^*)$ for a range of values of $\beta$ in Figure \ref{fig:beta}. It is found that as the model uncertainty increases, more loss will be transferred. This observation is consistent with the case when $\VaR$ is used as the risk measure; see Corollary \ref{cor:3}.  As $k$ decreases (or $\ell$ increases), $\Lambda_\beta(x^*)$ becomes larger and more risk is transferred. This is consistent with the case without model uncertainty. Moreover, when model uncertainty is large (i.e., $\beta$ is small), the DM risk attitude becomes less important to their optimal decision.

\begin{figure}[h!]
\centering
        \includegraphics[width=0.8\linewidth]{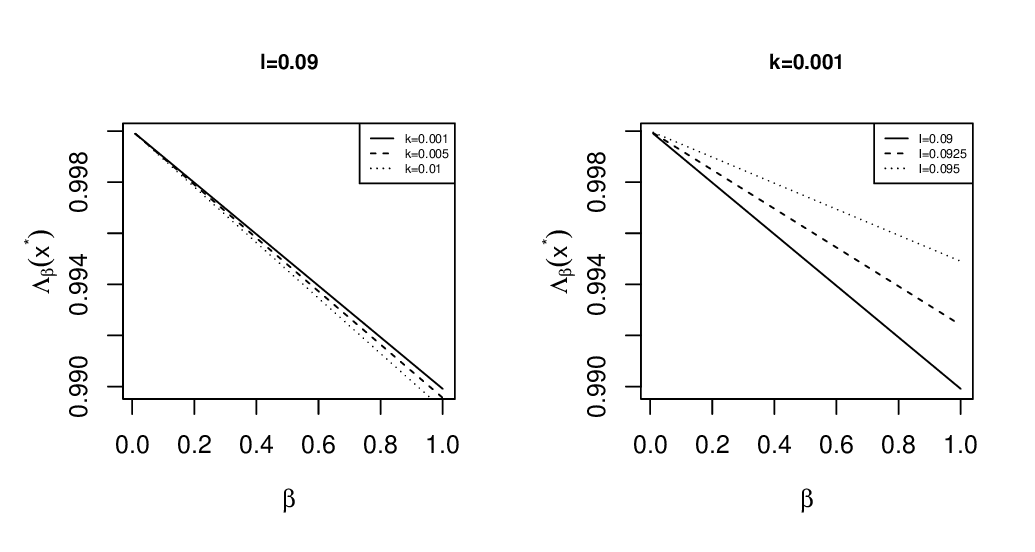}
\caption{$\Lambda_\beta(x^*)$  for a range of values of $\beta$ in problem \eqref{eq:robust1} under likelihood ratio uncertainty.}
        \label{fig:beta}
\end{figure}

Finally, we consider the case when $\Lambda'\VaR$ is incorporated in the premium principle (problem \eqref{eq:obj_3}).
By Theorem \ref{thm:3}, the optimal ceded loss function is $
f^{*}(x)= x \wedge  \VaR_{\Lambda(x^*)}(X)
$. We assume $\Lambda'=\Lambda$, and  plot $\Lambda(x^*)$ as a function of $\alpha\in (1.05,2.6)$ by fixing $\theta=0.25$. We choose to plot $\Lambda(x^*)$ on a different range of $\alpha$ from the range used in problem \eqref{eq:obj_general} mainly for better illustration. The observations are similar to the case of expected value premium as shown in Figure \ref{fig:alpha1}.

\begin{figure}[h!]
\centering
        \includegraphics[width=1\linewidth]{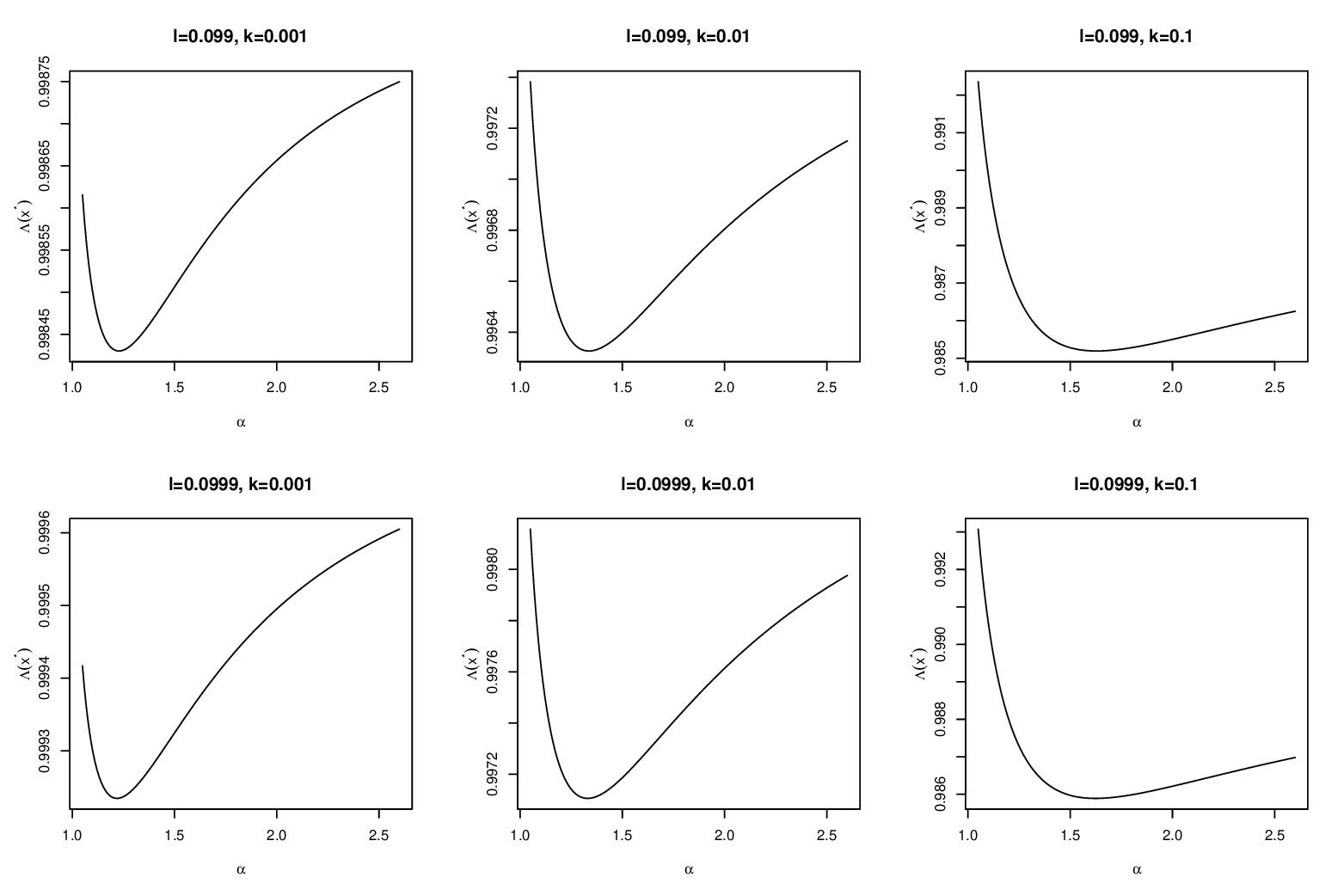}
     \caption{$\Lambda(x^*)$  for a range of values of $\alpha$ in problem \eqref{eq:obj_3}.}
        \label{fig:alpha1}
\end{figure}
}

\section{Conclusion}\label{sec:conc} This paper studies the objective of minimizing the $\Lambda\VaR$ of terminal losses for a DM who is considering to purchase insurance. The insurance premium is given by an expected value premium principle, or based on another $\Lambda'\VaR$ risk measure. Generally, a {limited} stop-loss indemnity or full/no insurance 
is shown to be optimal. If the DM only considers stop-loss indemnities, then the optimal deductible is shown in closed form when the premium is given by an expected value premium principle. Finally, this paper also considers the impact of model uncertainty.

One of the core techniques we use to derive the results is that we represent $\Lambda\VaR$ by $\VaR$ (see Lemma \ref{lem:representation} and Proposition \ref{prop:var_problem}). 
We convey the message that optimization problems based on $\Lambda\VaR$ (either maximization or minimization problems) can usually be safely transferred to equivalent $\VaR$ problems. As a result, the optimal solutions with $\Lambda\VaR$ and $\VaR$ seem to share similar forms. However, depending on specific practical needs, the former solution sometimes has several advantages over the latter, including that (i) the solution with $\Lambda\VaR$ is mathematically more general,  than that with $\VaR$ as its special case; (ii) the solution with $\Lambda\VaR$ allows the DM to adjust her risk/confidence level along with the model parameters and insurance prices. In light of these, optimal solutions (including their existence) for $\Lambda\VaR$ are also non-trivial to get (see in particular Theorems \ref{prop:2} and \ref{thm:3}).

While we believe that this study includes a wide class of cases, we conclude this paper with two suggestions for further research. First, this study assumes homogeneous beliefs, meaning that the premium principle and the objective function are based on the same probability measure. It would be interesting to relax this assumption, and let the premium principle be based on an alternative probability measure, just like in the risk-sharing problem of \cite{LTW24}. Second, we  assume in this paper that there is only one  insurer. For further research, we propose to study the impact of having multiple insurers in the market, each with a different premium principle. 

\vspace{1cm}
\noindent
{\large \bf Disclosure statement}

\vspace{0.2cm}
\noindent
 No potential conflict of interest was reported by the authors.
 
 \noindent
 
 \vspace{1cm}
\noindent{\large {\bf Acknowledgments}}
 \vspace{0.3cm}
\noindent

The authors would like to thank the editor-in-chief and anonymous referees for their careful reading and helpful comments on an earlier version of this paper, which leads to a considerable improvement of the presentation of the work.  The  research of  Xia Han  is supported by the National Natural Science Foundation of China (Grant Nos. 12301604, 12371471, and 12471449).

\appendix

\normalsize
     \setcounter{lemma}{0}
     \renewcommand{\thelemma}{A.\arabic{lemma}} 
          \setcounter{proposition}{0}
     \renewcommand{\theproposition}{A.\arabic{proposition}}
          \setcounter{definition}{0}
     \renewcommand{\thedefinition}{A.\arabic{definition}}
               \setcounter{corollary}{0}
     \renewcommand{\thecorollary}{A.\arabic{corollary}}
               \setcounter{example}{0}
     \renewcommand{\theexample}{A.\arabic{example}} 
               \setcounter{equation}{0}
     \renewcommand{\theequation}{A.\arabic{equation}} 
     
               \setcounter{figure}{0}
     \renewcommand{\thefigure}{A.\arabic{figure}}



\begin{thebibliography}{}
		\small
  \bibitem[\protect\citeauthoryear{Arrow}{Arrow}{1963}]{A63}
  Arrow, K.J. (1963). Uncertainty and the welfare economics of medical care. \emph{American Economic Review}, \textbf{53}(5), 941--973.
  

\bibitem[\protect\citeauthoryear{Asimit et al.}{Asimit et al.}{2017}]{ABCHK17}
Asimit, A.V., Bignozzi, V., Cheung, K.C., Hu, J. and Kim, E.S. (2017). Robust and Pareto optimality of insurance contracts. \emph{European Journal of Operational Research}, \textbf{262}(2), 720--732.

\bibitem[\protect\citeauthoryear{Asimit and Boonen}{Asimit and Boonen}{2018}]{AB18}
Asimit, A.V. and Boonen, T.J. (2018). Insurance with multiple insurers: A game-theoretic approach. \emph{European Journal of Operational Research}, \textbf{267}(2), 778--790.

  \bibitem[\protect\citeauthoryear{Assa}{Assa}{2015}]{A15}
Assa, H. (2015). On optimal reinsurance policy with distortion risk measures and premiums. \emph{Insurance: Mathematics and Economics}, \textbf{61}, 70--75.


\bibitem[\protect\citeauthoryear{Assa et al.}{Assa et al.}{2021}]{ASL21}
Assa, H., Sharifi, H. and Lyons, A. (2021). An examination of the role of price insurance products in stimulating investment in agriculture supply chains for sustained productivity. \emph{European Journal of Operational Research}, \textbf{288}(3), 918--934.


\bibitem[\protect\citeauthoryear{Balb\'as et al.}{2023}]{BBB23}
Balb\'as, A., Balb\'as, B. and Balb\'as, R. (2023). Lambda-quantiles as fixed points. \emph{SSRN:} 4583950.


\bibitem[\protect\citeauthoryear{Bellini and Peri}{Bellini and Peri}{2022}]{BP22}
 Bellini, F.~and Peri, I. (2022). An axiomatization of $\Lambda$-quantiles. \emph{SIAM Journal on Financial Mathematics}, \textbf{13}(1), 26--38.

\bibitem[\protect\citeauthoryear{Birghila et al.}{Birghila et al.}{2023}]{BBG23}
Birghila, C., Boonen, T. J. and  Ghossoub, M. (2023). Optimal insurance under maxmin expected utility. \emph{Finance and Stochastics}, \textbf{27}(2), 467--501.

\bibitem[\protect\citeauthoryear{Blanchet et al.}{Blanchet et al.}{2022}]{BCZ22}
Blanchet, J., Chen, L. and  Zhou, X. Y. (2022). Distributionally robust mean-variance portfolio selection with Wasserstein distances. \emph{Management Science}, \textbf{68}(9), 6382--6410.

 
\bibitem[\protect\citeauthoryear{Boonen and Jiang}{Boonen and Jiang}{2024}]{BJ24}
 Boonen, T.J. and Jiang, W. (2024). Robust insurance design with distortion risk measures. \emph{European Journal of Operational Research}, \textbf{316}(2), 694--706.

\bibitem[\protect\citeauthoryear{Borch}{Borch}{1960}]{B60}
Borch, K.K. (1960). An attempt to determine the optimum amount of stop loss reinsurance. In: \emph{Transactions of the 16 th International Congress of Actuaries I}, pp.~597--610.

\bibitem[\protect\citeauthoryear{Burzoni et al.}{Burzoni et al.}{2017}]{BPR17}
Burzoni, M., Peri, I. and Ruffo, C.M. (2017). On the properties of the Lambda value at risk: robustness, elicitability and consistency. \emph{Quantitative Finance}, \textbf{17}(11), 1735--1743.


  \bibitem[\protect\citeauthoryear{Cai and Chi}{2020}]{CC20}
Cai, J. and Chi, Y. (2020). Optimal reinsurance designs based on risk measures: A review. \emph{Statistical Theory and Related Fields}, \textbf{4}(1), 1--13.


\bibitem[\protect\citeauthoryear{Cai et al.}{2024}]{CLY24}
 Cai, J., Liu, F. and Yin, M. (2024). Worst-case risk measures of stop-loss and limited loss random variables under distribution uncertainty with applications to robust reinsurance. \emph{European Journal of Operational Research}, \textbf{318}(1), 310--326.
 
  \bibitem[\protect\citeauthoryear{Cai and Tan}{Cai and Tan}{2007}]{CT07}
  Cai, J. and Tan, K.S. (2007). Optimal retention for a stop-loss reinsurance under the VaR and CTE risk measures. \emph{ASTIN Bulletin}, \textbf{37}(1), 93--112.

      \bibitem[\protect\citeauthoryear{Cai et al.}{Cai et al.}{2008}]{CTWZ08}
  Cai, J., Tan, K.S., Weng, C., and Zhang, Y. (2008). Optimal reinsurance under VaR and CTE risk measures. \emph{Insurance: Mathematics and Economics}, \textbf{43}(1), 185--196.


  

\bibitem[\protect\citeauthoryear{Carlier and  Dana}{Carlier and  Dana}{2003}]{CD03}
  Carlier, G. and  Dana, R.A. (2003). Pareto efficient insurance contracts when the insurer's cost function is discontinuous. \emph{Economic Theory}, \textbf{21}, 871--893.

  
    


\bibitem[\protect\citeauthoryear{Chen and Wang}{2025}]{CW25}
Chen, Y.  and Wang, R. (2025).  Infinite-mean models in risk management: Discussions and recent advances. \emph{Risk Sciences}, 1:100003.


\bibitem[\protect\citeauthoryear{Cheung et al.}{Cheung et al.}{2014}]{CSYY14}
 Cheung, K.C., Sung, K.C.J., Yam, S.C.P. and  Yung, S.P. (2014). Optimal reinsurance under general law-invariant risk measures. \emph{Scandinavian Actuarial Journal}, \textbf{2014}(1), 72--91.

\bibitem[\protect\citeauthoryear{Chi and Tan}{Chi and Tan}{2011}]{CT11}
  Chi, Y. and Tan, K.S. (2011). Optimal reinsurance under VaR and CVaR risk measures: a simplified approach. \emph{ASTIN Bulletin},  \textbf{41}(2), 487--509.



\bibitem[\protect\citeauthoryear{Corbetta and Peri}{Corbetta and Peri}{2018}]{CP18}
 Corbetta, J. and Peri, I. (2018). Backtesting Lambda value at risk. \emph{European Journal of Finance}, \textbf{24}(13), 1075--1087.



\bibitem[\protect\citeauthoryear{Cui et al.}{Cui et al.}{2013}]{CYW13}
 Cui, W., Yang, J. and Wu, L. (2013). Optimal reinsurance minimizing the distortion risk measure under general reinsurance premium principles. \emph{Insurance: Mathematics and Economics}, \textbf{53}(1), 74--85.
  
  \bibitem[\protect\citeauthoryear{El Karoui and Ravanelli}{2009}]{KR09} El Karoui, N.  and  Ravanelli, C.  (2009). Cash subadditive risk measures and interest rate ambiguity. \emph{Mathematical Finance}, \textbf{19}(4), 562--590.
  \bibitem[\protect\citeauthoryear{Embrechts et al.}{Embrechts et al.}{2018}]{ELW18}
Embrechts, P., Liu, H. and  Wang, R. (2018). Quantile-based risk sharing. \emph{Operations Research}, \textbf{66}(4), 936--949.

 \bibitem[\protect\citeauthoryear{F\"ollmer and Schied}{F\"ollmer and Schied}{2016}]{FS16} F\"ollmer, H.~and Schied, A.~(2016). \emph{Stochastic Finance. An Introduction in Discrete Time}. Fourth Edition.  {Walter de Gruyter, Berlin}.


\bibitem[\protect\citeauthoryear{Frittelli et al.}{Frittelli et al.}{2014}]{FMP14}
Frittelli, M., Maggis, M. and Peri, I. (2014). Risk measures on $\mathcal P(\mathbb R)$ and value at risk with probability/loss function. \emph{Mathematical Finance}, \textbf{24}(3), 442--463.

\bibitem[\protect\citeauthoryear{Glasserman and Xu}{Glasserman and Xu}{2013}]{GX13} Glasserman, P. and Xu, X. (2013). Robust portfolio control with stochastic factor dynamics. \emph{Operations Research}, \textbf{61}(4), 874--893.


  \bibitem[\protect\citeauthoryear{Han and Liu}{Han and Liu}{2024}]{HL24}
Han, X. and  Liu, P. (2024). Robust Lambda-quantiles and extreme probabilities. \emph{arXiv}: 2406.13539. 

  \bibitem[\protect\citeauthoryear{Han et al.}{Han et al.}{2025}]{HWWX25} Han, X., Wang, Q., Wang, R. and Xia, J. (2025). Cash-subadditive risk measures without quasi-convexity. \emph{Mathematics of Operations Research}, forthcoming.

   
\bibitem[\protect\citeauthoryear{Hitaj et al.}{Hitaj et al.}{2018}]{HMP18}
  Hitaj, A., Mateus, C. and Peri, I. (2018). Lambda value at risk and regulatory capital: a dynamic approach to tail risk. \emph{Risks}, \textbf{6}(1), 17.

  \bibitem[\protect\citeauthoryear{Huang et al.}{Huang et al.}{2010}]{HZFF10} 
Huang, D., Zhu, S., Fabozzi, F. J. and Fukushima, M. (2010). Portfolio selection under distributional uncertainty: A relative robust CVaR approach. \emph{European Journal of Operational Research}, \textbf{203}(1), 185--194.

  

\bibitem[\protect\citeauthoryear{Huberman et al.}{Huberman et al.}{1983}]{HMS83}
Huberman, G., Mayers, D., Smith, C.W., Jr. (1983). Optimal insurance policy indemnity schedules. \emph{The Bell Journal of Economics}, \textbf{14}(2), 415--26.

\bibitem[\protect\citeauthoryear{Ince et al.}{Ince et al.}{2022}]{IPP22}
Ince, A., Peri, I. and Pesenti, S. (2022). Risk contributions of lambda quantiles. \emph{Quantitative Finance}, \textbf{22}(10), 1871--1891.

\bibitem[\protect\citeauthoryear{Kaas et al.}{Kaas et al.}{2008}]{KGDD08} Kaas, R., Goovaerts, M., Dhaene, J. and Denuit, M. (2008). \textit{Modern Actuarial Risk Theory: Using R}. \emph{Springer} (2nd edition), Heidelberg, Germany. 

\bibitem[\protect\citeauthoryear{Klages-Mundt and Minca}{Klages-Mundt and Minca}{2020}]{KM20}
Klages-Mundt, A. and  Minca, A. (2020). Cascading losses in reinsurance networks. \emph{Management Science}, \textbf{66}(9), 4246--4268.

\bibitem[\protect\citeauthoryear{Landriault et al.}{2024}]{LLS24}
Landriault, D., Liu, F. and Shi, Z. (2024). Worst-case reinsurance strategy with likelihood ratio uncertainty. \emph{SSRN}: 4750381.

	\bibitem[\protect\citeauthoryear{Liu et al.}{2020}]{LCLW20}
 Liu, F., Cai, J., Lemieux, C. and Wang, R. (2020). Convex risk functionals: Representation and applications.
 \emph{Insurance: Mathematics and Economics},  \textbf{90}, 66--79.

 \bibitem[\protect\citeauthoryear{Liu et al.}{2022}]{LMWW22}
Liu, F., Mao, T., Wang, R. and Wei, L. (2022). Inf-convolution, optimal allocations, and model uncertainty for tail risk measures. \emph{Mathematics of Operations Research}, \textbf{47}(3), 2494--2519.

	\bibitem[\protect\citeauthoryear{Liu and Mao}{2022}]{LM22}
 Liu, H. and  Mao, T. (2022). Distributionally robust reinsurance with Value-at-Risk and Conditional Value-at-Risk. \emph{Insurance: Mathematics and Economics}, \textbf{107}, 393--417.
 

 \bibitem[\protect\citeauthoryear{Liu}{Liu}{2024}]{L24}
Liu, P. (2024). Risk sharing with Lambda value at risk. \emph{Mathematics of Operations Research}, doi.org/10.1287/moor.2023.0246.

 \bibitem[\protect\citeauthoryear{Liu et al.}{2024}]{LTW24}
Liu, P., Tsanakas, A. and Wei, Y. (2024). Risk sharing with Lambda value at risk under heterogeneous beliefs. \emph{arXiv}: 2408.03147.

\bibitem[\protect\citeauthoryear{Natarajan et al.}{2010}]{NSU10}
Natarajan, K., Sim, M., and Uichanco, J. (2010). Tractable robust expected utility and risk models for portfolio optimization. \emph{Mathematical Finance}, \textbf{20}(4), 695--731.

\bibitem[\protect\citeauthoryear{Pesenti  et al.}{2024}]{PWW20}
Pesenti, S., Wang, Q. and Wang R. (2024). Optimizing distortion riskmetrics with distributional
uncertainty. \emph{Mathematical Programming}, forthcoming.

\bibitem[\protect\citeauthoryear{Shao and Zhang}{2023}]{SZ23}
Shao, H. and Zhang, Z.G. (2023)
Distortion risk measure under parametric ambiguity. \emph{European Journal of Operational Research}, \textbf{331}(3), 1159--1172.

\bibitem[\protect\citeauthoryear{Tan et al.}{2020}]{TWWZ20}
Tan, K.S., Wei, P., Wei, W. and Zhuang, S.C. (2020). Optimal dynamic reinsurance policies under a generalized Denneberg’s absolute deviation principle. \emph{European Journal of Operational Research}, \textbf{282}(1), 345--362.


 \bibitem[\protect\citeauthoryear{Xia and Hu}{Xia and Hu}{2024}]{XH24}
Xia, Z. and Hu, T. (2024). Optimal risk sharing for lambda value-at-risk. \emph{Advances in Applied Probability}, doi.org/10.1017/apr.2024.27.


 \bibitem[\protect\citeauthoryear{Zhuang et al.}{Zhuang et al.}{2016}]{ZWTA16}
Zhuang, S. C., Weng, C., Tan, K. S. and Assa, H. (2016). Marginal indemnification function formulation for optimal reinsurance. \emph{Insurance: Mathematics and Economics}, \textbf{67}, 65--76.
	\end{thebibliography}
	\end{document}